\documentclass[draft]{IEEEtran}
\usepackage[small,nohug,heads=vee]{diagrams}
\diagramstyle[labelstyle=\scriptstyle]
\usepackage[english]{babel}
\usepackage{amsmath,amsthm,mathrsfs}
\usepackage{amsfonts}
\usepackage{extarrows}
\usepackage{amssymb}
\usepackage{dblfloatfix}
\usepackage[numbers]{natbib}
\usepackage[final,hyperfootnotes=false]{hyperref}
\usepackage[final]{graphicx}
\usepackage{tikz}
\usepackage{pgfplots}
\usepackage{subfig}
\usepackage{flushend}

\usetikzlibrary{arrows,positioning,fit,backgrounds,shapes}
\usetikzlibrary{calc}
\onecolumn
\newtheorem{thm}{Theorem}

\newtheorem{cor}{Corollary}
\newtheorem{lem}{Lemma}
\newtheorem{prop}{Proposition}
\theoremstyle{definition}
\newtheorem{defn}{Definition}

\theoremstyle{remark}
\newtheorem{rem}{Remark}

\newtheorem{conjecture}{Conjecture}

\DeclareMathOperator*{\diag}{diag}
\DeclareMathOperator*{\env}{env}
\DeclareMathOperator*{\xenv}{env_X}
\DeclareMathOperator*{\yenv}{env_Y}
\DeclareMathOperator*{\xyenv}{env_{XY}}
\newcommand{\mb}{\mathbf}
\newcommand{\rd}{\rm{d}}
\newcommand{\half}{\frac{1}{2}}

\DeclareMathOperator*{\supp}{supp}

\newcommand\blfootnote[1]{%
  \begingroup
  \renewcommand\thefootnote{}\footnote{#1}%
  \addtocounter{footnote}{-1}%
  \endgroup
}

\usepackage{xcolor}
\hypersetup{
    colorlinks,
    linkcolor={blue!80!black},
    citecolor={blue!80!black},
    urlcolor={blue!80!black}
}

\begin{document}
\title{Secret Key Generation with Limited Interaction
}
\author{Jingbo Liu\quad Paul Cuff\quad Sergio Verd\'{u}\\
Dept. of Electrical Eng., Princeton University, NJ 08544\\
\{jingbo,cuff,verdu\}@princeton.edu}%
\maketitle
\begin{abstract}
A basic two-terminal secret key generation model is considered, where the interactive communication rate between the terminals may be limited, and in particular may not be enough to achieve the maximum key rate.
We first prove a multi-letter characterization of the key-communication rate region (where the number of auxiliary random variables depend on the number of rounds of the communication),
and then provide an equivalent but simpler characterization in terms of concave envelopes in the case of unlimited number of rounds.
Two extreme cases are given special attention.
First, in the regime of very low communication rates, the \emph{key bits per interaction bit} (KBIB) is expressed with a new ``symmetric strong data processing constant'',
which has a concave envelope characterization analogous to that of the conventional strong data processing constant.
The symmetric strong data processing constant can be upper bounded by
the supremum of the maximal correlation coefficient over a set of distributions,
which allows us to determine the KBIB for binary symmetric sources,
and conclude, in particular, that the interactive scheme is not more efficient than the one-way scheme at least in the low communication-rate regime.
Second, a new characterization of the \emph{minimum interaction rate needed for achieving the maximum key rate} (MIMK) is given, and we resolve a conjecture by Tyagi regarding the MIMK for (possibly nonsymmetric) binary sources.
We also propose a new conjecture for binary symmetric sources that the interactive scheme is not more efficient than the one-way scheme at any communication rate.
\end{abstract}

\blfootnote{This paper was presented in part at the 2016 IEEE International Symposium on Information Theory (ISIT) Barcelona, July 10-15, 2016.
}

\IEEEpeerreviewmaketitle
\section{Introduction}
Generally speaking,
secret key generation
\cite{maurer1993secret}\cite{ahlswede1993common}\cite{csiszar2004secrecy} concerns the task of producing a common piece of information by several terminals accessing dependent sources, possibly allowing communications among the terminals,
so that an eavesdropper, knowing the sources joint distribution and the protocol, and observing the communications but not the source outputs, can learn almost nothing about the common information generated.
The importance of key generation in cryptography and other areas of information theory is well known \cite{shannon1949communication}\cite{maurer1993secret}\cite{ahlswede1993common}\cite{ahlswede1998common}.
From a more theoretical viewpoint,
secret key generation is also a rich playground,
because of its connections to various measures of dependence.
Consider the case of two i.i.d.~sources with per-letter distribution $Q_{XY}$.
The maximum key rate that can be produced
without any constraint on the communication rate equals the mutual information $I({X;Y})$ \cite{ahlswede1993common}. In the other extreme where the communication rate vanishes, the \emph{key bits per communication bit} under the one-way protocol in \cite{ahlswede1993common} is a monotonic function of the strong data processing constant (see e.g.\ \cite{courtade2013outer}\cite{Liu}
);
and under the omniscient helper protocol \cite{liu2015key}, the region of the achievable communication bits per key bit (which is a vector of dimension equal to the number of receivers) is a reflection of the \emph{polar set} (see the definition in, e.g.\ \cite[P125]{rockafellar1970convex} or \cite[Section~3.b.]{milman2007geometrization}) of the set of hypercontractive coefficients (which can be seen \cite[(44)]{liu2015key}
and \cite[Remark~8]{liu2015key}).

Despite the successes mentioned above, which mainly
concern models allowing only
unidirectional communication among terminals, many basic problems have remained open in settings involving interactive communications or multi-terminals \cite{el2011network}\cite{csiszar2000common}\cite{csiszar2004secrecy}.
Most of the existing literature
focuses
on the minimum communication rate
required to achieve
the maximum possible key rate
(the most notable exception is the basic source model with one-way communication, where the complete tradeoff between the key rate and the communication rate is known; see \cite{ahlswede1993common}\cite{watanabe2010}).
Csisz\'{a}r and Narayan \cite{csiszar2004secrecy} showed that the maximum key rate obtainable from multi-terminals having public interactive communications equals the entropy rate of all sources minus the communication rate needed for \emph{communication for omniscience}
\cite{csiszar2004secrecy}, the latter related to the subject of interactive source coding studied by Kaspi \cite{kaspi1985two}.
Moreover, Tyagi \cite{tyagi2013common} showed that the \emph{minimum interactive communication rate needed for achieving the maximum key rate} (MIMK) between two interacting terminals equals the \emph{interactive common information} \cite{tyagi2013common} minus the mutual information rate of the sources, and provided a multi-letter characterization of MIMK.
However, a complete characterization of the tradeoff between the key rate and communication rate is more challenging,
because when the communication rate is not large enough for the terminals to become omniscient, it not obvious what piece of information they have to agree on. Indeed, as mentioned at the end of Section~VII in \cite{tyagi2013common}, a characterization of the key rate when the communication rate is less than MIMK, along with a single-letter characterization of MIMK, is an interesting open problem.
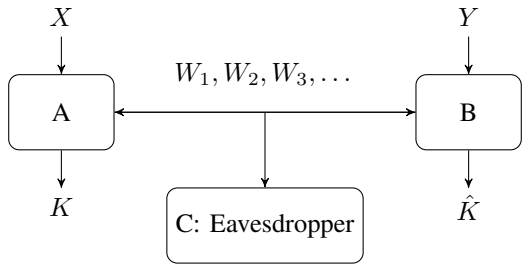
\begin{figure}[ht]
  \centering
\begin{tikzpicture}
[node distance=1cm,minimum height=10mm,minimum width=14mm,arw/.style={->,>=stealth'}]
  \node[coordinate] (p1) {};
  \node[rectangle,draw,rounded corners] (C) [below =of p1] {C: Eavesdropper};
  \node[rectangle,draw,rounded corners] (A) [left =2cm of p1] {A};
  \node[rectangle,draw,rounded corners] (B) [right =2cm of p1] {B};
  \node[coordinate] (X) [above =0.5cm of A] {};
  \node[coordinate] (Y) [above =0.5cm of B] {};
  \node[coordinate] (K) [below =0.5cm of A] {};
  \node[coordinate] (Kh) [below =0.5cm of B] {};

  \draw [arw] (X) to node[midway,above]{$X$} (A);
  \draw [arw] (Y) to node[midway,above]{$Y$} (B);
  \draw [arw] (A) to node[midway,below]{$K$} (K);
  \draw [arw] (B) to node[midway,below]{$\hat{K}$} (Kh);
  \draw [arw] (A) to node[midway,above]{$W_1,W_2,W_3,\dots$} (B);
  \draw [arw] (B) to node[midway,below]{~} (A);
  \draw [arw] (p1) to node{} (C);
\end{tikzpicture}
\caption{\small Terminals A and B observe sources with joint distribution $Q_{XY}$ and interactive (two-way) communication between A and B is allowed.}
\label{figinteractive}
\end{figure}

In this paper we consider the two-terminal interactive key generation model shown in Figure~\ref{figinteractive}, which is similar to the setting of \cite{tyagi2013common}. However, instead of examining only the MIMK, we analyzed the full tradeoff between key rate and communication rate using a different approach.

In Section~\ref{sec_pre}, we revisit Kaspi's original idea of
multi-letter characterizations of the rate region of interactive source coding \cite{kaspi1985two}, where each round of communication accounts for a new auxiliary random variable and
introduces an additional term
to the rate expressions which resembles the expressions in the one-way counterparts.
In our interactive key generation problem, we derive a similar multi-letter characterization as the first step.

In Section~\ref{sec_convgeo}, we then simplify the multi-letter characterization of the key-communication tradeoff region using \emph{$XY$-concave envelopes}, partially inspired by a similar characterization in the context of interactive source coding by Ma et~al.~\cite{ma2012interactive},
who noticed that each auxiliary random variable in the multi-letter region, which corresponds to each
public communication round,
amounts to convexifying the rate region with respect a marginal distribution. Hence in the infinite-round limit, the minimum sum rate can be described in terms of a marginally convex envelope, i.e.,the the greatest functional which is convex with respect to each marginal distribution and dominated by a given functional.
The convexifying role of auxiliary random variables is
very common in information theory
(e.g.~\cite{el2011network})
At first sight, this idea is easily overlooked as a mere restatement of the multi-letter region. However, as demonstrated by Ma et~al.'s work \cite{ma2012interactive} as well as the present paper, the conceptual simplification opens the possibility of tackling specific open questions and making new connections.
Moreover, we introduce a notions of $X$-absolute continuity
and $XY$-absolute continuity, so that the marginal concave envelope approach is applicable to general non-discrete sources. In fact this framework may be applied to other problems involving convex/concave envelopes to avoid the technical difficulty of defining a conditional distribution from a given joint distribution.

Section~\ref{sec_kbib} focuses on the regime of very small communication rates.
The \emph{key bits per interaction bit} (KBIB) is defined as the fundamental limit on the maximum amount of key bits that can be ``unlocked'' by each communication bit, which is the most befitting for the scenario of stringent communication constraints and relatively abundant correlated resources.
Generally, KBIB is not completely implied by the rate region since the length of the communication bits can be a vanishing fraction of the blocklength;
but for stationary memoryless sources, it can be shown to be equal to the slope of the boundary of the achievable rate region at the origin.
We introduce a ``symmetric strong data processing constant'' (SSDPC), defined as the minimum of a parameter such that a certain information-theoretic functional touches its $XY$-concave envelope at a given source distribution.
The concave envelope characterization of the achievable rate region for key generation developed in Section~\ref{sec_convgeo} can then be used to show that KBIB
is a monotonic function of SSDPC.
It is interesting to compare SSDPC with the conventional strong data processing constant \cite{ahlswede1976spreading}, which has a similar (but only one-sided) convex envelope characterization \cite{anantharam2013maximal} and is similarly related to the key bits per communication bit in the one-way protocol \cite{courtade2013outer}\cite{Liu}.
For binary symmetric sources and  Gaussian sources,
we show that the SSDPC coincides with the strong data processing constant, implying that one-way communication is sufficient for achieving KBIB for those sources.

Returning to the MIMK problem considered by Tyagi \cite{tyagi2013common},
a new characterization of the minimum interaction rate needed for achieving the key capacity is given.
The proof of this characterization relies on a saddle
point property of an optimization problem,
which we prove in the finite-alphabet case by establishing semicontinuity and convergence properties of $XY$-concave functions.
In \cite{tyagi2013common} Tyagi conjectured that MIMK equals the minimum one-round communication rate for achieving the maximum key rate for all binary sources.
Here we use the new characterization of MIMK to prove that,
unless the binary random variables are independent, the necessary and sufficient condition for the conjecture to hold is that the conditional distribution of one binary random variable given the other is a binary symmetric channel.
We also propose a new conjecture that the complete key-interaction tradeoff region for binary symmetric sources can be achieved with one-way communication, and provide some numerical and analytical evidence for its validity.

\section{preliminary}\label{sec_pre}
\subsection{Problem Setup}
In Figure~\ref{figinteractive},
let $Q_{XY}$ be the joint distribution of the sources.
Terminals A and B observe $X$ and $Y$, respectively.
Terminal A computes an integer $W_1=W_1(X)$ (possibly stochastically) and sends it to B.
Then B computes an integer $W_2=W_2(W_1,Y)$ and sends it to A, and so on, for a total of $r$ rounds.
Then, A and B calculate the integer-valued keys\footnote{Notation $W_i^j:=(W_i,W_{i+1},\dots,W_j)$ denotes a vector and $W^r:=W_1^r$.} $K=K(X,W^r)$ and $\hat{K}=\hat{K}(Y,W^r)$, respectively, possibly stochastically.
The objective is that $K=\hat{K}$ with high probability, and that $K$ is (almost) independent of the public messages $W^r$ observed by the eavesdropper.
The compliance with these two objectives can be measured by a single quantity (see for example \cite{Tyagi2015}\cite{hayashi2016}):
\begin{align}
\Delta_n:=\frac{1}{2}|Q_{K\hat{K}W^r}-T_{K\hat{K}}Q_{W^r}|.
\label{e_mdelta}
\end{align}
Here $T_{K\hat{K}}$ denotes the target distribution under which $K=\hat{K}$ is equiprobable, that is,
\begin{align}
T_{K\hat{K}}(k,\hat{k}):=\frac{1}{|\mathcal{K}|}
1\{k=\hat{k}\},\quad\forall k,\hat{k}\in\mathcal{K}.
\label{e_tk}
\end{align}
The total variation $|\cdot|$ is defined as the $\ell_1$ distance.
Note that such a performance measure arises naturally when the \emph{likelihood encoder} is used in the achievability proof (see for example \cite{Liu} or Appendix~\ref{app_region}).

In the case of stationary memoryless sources and block coding, we substitute $X\leftarrow X^n$ and $Y\leftarrow Y^n$, where $n$ is the blocklength.

\begin{defn}
The triple $(R,R_1,R_2)$ is said to be \emph{achievable in $r$ rounds} ($r\in\{1,2,\dots,\infty\}$) if a sequence of generation schemes (indexed by the blocklength $n$) in $r$ rounds\footnote{As a convention, we say ``in $r$ rounds'' or ``$r$-round'' if the number of rounds of communication less than $r+1$ (or equivalently, not exceeding $r$ for an integer $r$ or finite for $r=\infty$). Therefore the term is not precise if $r=\infty$.}
can be designed to fulfill the following conditions:
\begin{align}
\liminf_{n\to\infty}\frac{1}{n}\log|\mathcal{K}|&\ge R;
\\
\limsup_{n\to\infty}
\sum_{l\in\mathcal{O}^r}\frac{1}{n}\log|\mathcal{W}_l|&\le R_1;
\\
\limsup_{n\to\infty}
\sum_{l\in\mathcal{E}^r}\frac{1}{n}\log|\mathcal{W}_l|&\le R_2;
\\
\lim_{n\to\infty}\Delta_n&=0.
\end{align}
where we used the following notation: for integers $s$ and $r$,
\begin{align}
\mathcal{E}_s^r&:=\{s,2,\dots,r\}\cap 2\mathbb{Z};
\\
\mathcal{O}_s^r&:=\{s,2,\dots,r\}\setminus 2\mathbb{Z},
\end{align}
and $\mathcal{E}_1^r$ and $\mathcal{O}_1^r$ are abbreviated as $\mathcal{E}^r$ and $\mathcal{O}^r$, respectively.
\end{defn}
\begin{rem}\label{rem_equiv}
Some authors (see e.g.~\cite[(6)]{csiszar2004secrecy}) have considered the following alternative performance metrics
\begin{align}
\epsilon_n&:=\mathbb{P}[K\neq \hat{K}],
\label{e_m0}
\\
\nu_n&:=\max\{D(Q_{K|W^r}\|T_K|Q_{W^r}),
\,D(Q_{\hat{K}|W^r}\|T_{\hat{K}}|Q_{W^r})\}.
\label{e_m1}
\end{align}
The relation to \eqref{e_mdelta} is as follows:
Clearly, $\Delta_n\to0$ implies that $\epsilon_n\to0$. Also, notice that for arbitrary $P$ and $Q$ on the same alphabet $\mathcal{X}$, \cite[Lemma~2.7]{csiszar2011information} gives
\begin{align}
|H(P)-H(Q)|&\le |P-Q|\log\frac{|\mathcal{X}|}{|P-Q|}+\log|\mathcal{X}|1\left\{
|P-Q|>\frac{1}{2}\right\}
\end{align}
Thus by Jensen's inequality and Markov inequality, we have
\begin{align}
D(Q_{K|W^r}\|T_K|Q_{W^r})&=H(T_K)-H(Q_{K|W^r}|Q_{W^r})
\\
&\le 2\Delta_n\log\frac{|\mathcal{K}_n|}{2\Delta_n}
+4\Delta_n\log|\mathcal{K}_n|.
\end{align}
The same upper-bound holds for $D(Q_{\hat{K}|W^r}\|T_{\hat{K}}|Q_{W^r})$.
Therefore, if we assume that $|\mathcal{K}_n|=e^{O(n)}$, then $\Delta_n=o(1)$ ensures that $\nu_n=o(n)$.
On the other hand, by Pinsker's
inequality, $\nu_n\to0$ implies $|Q_{KW^r}-T_KQ_{W^r}|\to0$, which, combined with $\epsilon_n\to 0$, implies that $\Delta_n\to 0$.
\end{rem}

\begin{defn}
The set of achievable rate tuples for key generation in $r$ ($r\in\{0,1,2,\dots,\infty\}$) rounds is denoted by $\mathcal{R}_r(X,Y)$.
Define the \emph{total sum rate}
\begin{align}
S:=R+R_1+R_2,
\end{align}
and denote by $\mathcal{S}_{r}(X,Y)$ the set of achievable $(S,R)$.
Note that in the $r=0$ case, $\mathcal{R}_0(X,Y)=\mathcal{S}_0(X,Y)
=\emptyset$.
\end{defn}

From the standard diagonalization argument~\cite{han_s}, $\mathcal{R}_r(X,Y)$ and $\mathcal{S}_{r}(X,Y)$ are closed.

Clearly $\mathcal{R}_r(X,Y)$ is ``increasing'' in $r$. We can also show that it is ``continuous'' at $r=\infty$, that is
$\mathcal{R}_{\infty}(X,Y)$ equals the closure of $\bigcup_{r=1}^{\infty}\mathcal{R}_r(X,Y)$.
The ``$\supseteq$'' part is immediate from the definitions.
Although less obvious, ``$\subseteq$'' can be seen from Theorem~\ref{thm_region}.

The set $\mathcal{S}_{r}(X,Y)$ is a linear transformation of $\mathcal{R}_{r}(X,Y)$;
the former sometimes admits simpler expressions than the latter, but does not distinguish the communication rates in the two directions.

Inspired by Kaspi's multi-letter characterization of the rate region for interactive source coding \cite{kaspi1985two}, we prove in Appendix~\ref{app_region} and \ref{app_conv} that
\begin{thm}\label{thm_region}
For stationary memoryless sources
and $r\in\{0,1,2,\dots,\infty\}$,
$\mathcal{R}_r(X,Y)$
is the closure of the set of $(R,R_1,R_2)$ satisfying
\begin{align}
R\le& \sum_{i\in\mathcal{O}^r} I(U_i;Y|U^{i-1})
+\sum_{i\in\mathcal{E}^r} I(U_i;X|U^{i-1}),\label{e_rregion1}
\\
R_1\ge&\sum_{i\in\mathcal{O}^r}I(U_i;X|U^{i-1})
-\sum_{i\in\mathcal{O}^r}I(U_i;Y|U^{i-1}),\label{e_rregion2}
\\
R_2\ge&\sum_{i\in\mathcal{E}^r}I(U_i;Y|U^{i-1})
-\sum_{i\in\mathcal{E}^r}I(U_i;X|U^{i-1}),\label{e_rregion3}
\end{align}
and consequently
$\mathcal{S}_r(X,Y)$
is the closure of the set of $(S,R)$ satisfying
\begin{align}
S\ge& \sum_{i\in\mathcal{O}^r} I(U_i;X|U^{i-1})
+\sum_{i\in\mathcal{E}^r} I(U_i;Y|U^{i-1}),
\\
R\le& \sum_{i\in\mathcal{O}^r} I(U_i;Y|U^{i-1})
+\sum_{i\in\mathcal{E}^r} I(U_i;X|U^{i-1}),
\end{align}
where in both cases the auxiliary r.v.'s satisfy
\begin{align}
U_i-(X,U^{i-1})-Y, &\quad \textrm{odd }i\in\{1,2,\dots,r\};\label{e_markov1}
\\
X-(Y,U^{i-1})-U_i,  &\quad \textrm{even }i\in\{1,2,\dots,r\}.\label{e_markov2}
\end{align}
\end{thm}
\begin{rem}
For finite $X$ and $Y$, the achievability part of Theorem~\ref{thm_region} can also be obtained from \cite[Theorem~3]{Gohari2012} by setting
$Z=S = Y_1 = Y_2 = \emptyset$ and $R_0 = 0$ (see \cite[Remark~3]{Gohari2012}, but with finite $R_{12}$ and $R_{21}$) which was established by the \emph{output statistics of random binning} technique.
Our present proof in Appendix~\ref{app_region},
based on the \emph{likelihood encoder} \cite{song},
directly applies to general distributions.
\end{rem}
\begin{rem}\label{rem_finite}
Using the fact that the mutual information equals its supremum over finite partitions
(see e.g.\ \cite[Proposition 3.2.3]{polyanskiy2016lecture}),
we can show that the $\mathcal{U}_1,\dots,\mathcal{U}_r$ in Theorem~\ref{thm_region}
can be restricted to be finite,
even when $\mathcal{X}$ and $\mathcal{Y}$ are not finite.
The boundary points in $\mathcal{R}_r(X,Y)$ may not be equal to the quantities on the right sides of \eqref{e_rregion1}-\eqref{e_rregion3} for some finite $U^r$, but can be approximated by choosing a sequence of finite $U^r$.
\end{rem}
This bound is quite intuitive: depending on whether $i$ is odd or even, $U_i$ corresponds to the messages sent by either the terminal A or B. The first round of communication contributes to the term
$
(I(U_1;Y),I(U_1;X)-I(U_1;Y),0)
$
in the rate tuple expressions, which are exactly the rates in one-round key generation \cite{ahlswede1993common}.
The second round
contributes similar mutual informations except that they
are now conditioned on $U_1$, which is now shared publicly, and so on.

The related common randomness (CR) generation problem is similar to the key generation problem,
except that $K$ need not be almost independent of $W^r$.
For stationary memoryless sources,
the achievable region for common randomness generation is a linear transform of the achievable region for key generation;
see e.g.~the arguments in \cite{ahlswede1998common}\cite{tyagi2013common}.
For a given allowable interactive communication rate,
the maximum achievable CR rate is the maximum achievable key rate in the corresponding key generation problem
plus the allowable interactive communication rate (when local randomization is allowed).

\subsection{$XY$-Absolutely Continuity}\label{sec_abs}
In this subsection, we introduce a framework for convex geometric representations of rate regions, which will later (Section~\ref{sec_convgeo}) be applied to the key generation problem.
The convex geometric representation is closely related to representations via auxiliary random variables.
Consider a distribution $Q_{UXY}$ under which $U-X-Y$, and suppose that $|\mathcal{U}|<\infty$.
Then for any $u\in\mathcal{U}$,
\begin{align}
\frac{{\rm d}Q_{XY|U=u}}
{{\rm d}Q_{XY}}
=\frac{Q_{U|X=x}(u)}{Q_U(u)}
\end{align}
where $\frac{Q_{U|X=\cdot}(u)}{Q_U(u)}\le \frac{1}{\min_u Q_U(u)}$ is a bounded function depending on $x$ only.
Moreover, consider $Q_{U^rXY}$ under which the Markov chains \eqref{e_markov1} and \eqref{e_markov2} are satisfied,
where $U^r$ can be assumed finite (Remark~\ref{rem_finite}).
By repeating the same argument, we see that for any $u_1,\dots,u_r$, there exist bounded functions $f\colon \mathcal{X}\to [0,\infty)$ and $g\colon\mathcal{Y}\to [0,\infty)$ such that
\begin{align}
\frac{{\rm d}Q_{XY|U^r=u^r}}{{\rm d}Q_{XY}}=fg.
\end{align}
These observations motivate the following definition:
\begin{defn}\label{defn_abs}
A nonnegative finite measure $\nu_{XY}$ is said to be
\emph{$XY$-absolutely continuous} with respect to $\mu_{XY}$, denoted as $\nu_{XY}\preceq\mu_{XY}$, if there exists bounded\footnote{The boundedness assumption gives certain technical conveniences; for example the measure in \eqref{e_theta1} is guaranteed to be finite.} measurable
functions $f$ and $g$ on $\mathcal{X}$ and $\mathcal{Y}$, respectively, such that
\begin{align}
\frac{{\rm d}\nu_{XY}}{{\rm d}\mu_{XY}}=fg,
\label{e_xyabs}
\end{align}
$\mu_{XY}$-almost surely.
Further $\nu_{XY}$ is
said to be \emph{$X$-absolutely continuous} with respect to $\mu_{XY}$, denoted as
\begin{align}
\nu_{XY}\preceq_X \mu_{XY}
\label{e_xabs}
\end{align}
if one can take $g=1$ in \eqref{e_xyabs}.
 \end{defn}

\begin{rem}
It is straightforward to see that $\nu_{XY}\preceq \mu_{XY}$ if there exist measures $(\theta_{XY}^i)_{i=1}^t$ for some odd integer $t$ such that
\begin{align}
\nu_{XY}&\preceq_Y\theta_{XY}^t;
\\
\theta_{XY}^i&\preceq_X \theta_{XY}^{i-1},\quad i\in\{1,\dots,t\}\setminus2\mathbb{Z};
\\
\theta_{XY}^i&\preceq_Y \theta_{XY}^{i-1},\quad i\in\{1,\dots,t\}\cap2\mathbb{Z};
\\
\theta_{XY}^1&\preceq_X\mu_{XY}.
\end{align}
The converse is also true, and one can in fact choose $t=1$: consider $\mu_{XY}$, $\nu_{XY}$, $f$ and $g$ as in \eqref{e_xyabs},
then put
\begin{align}
{\rm d}\theta_{XY}^1&:=f{\rm d}\mu_{XY}
\label{e_theta1}
\end{align}
which is guaranteed to be a finite measure since $f$ is assumed to be bounded.
\end{rem}

\begin{defn}
The relation $\preceq_X$ is a \emph{preorder} relation
\footnote{A preorder relation satisfies reflexivity and transitivity, but not necessarily antisymmetry. The more familiar notion of partially ordered set is a preordered set satisfying antisymmetry.}
on the set of nonnegative finite measures. We denote by
\begin{align}
\mathcal{M}_X(\mu):=\{\nu\colon \nu\preceq_X\mu\}
\end{align}
the \emph{lower set} of $\mu$ in the set of nonnegative finite measures. Similarly, $\mathcal{M}(\mu)$ is defined as the lower set of $\mu$ with respect to $\preceq$. Both relations also make the set of probability distributions a preordered set. Denote by $\mathcal{P}_X(Q_{XY})$ and $\mathcal{P}(Q_{XY})$ the corresponding lower sets.
\end{defn}
\begin{rem}
The lower set $\mathcal{P}(Q_{XY})$ appears frequently information theory (with different notations and names). Csisz\'{a}r \cite{csiszar1975divergence} showed that the $I$-projection of $Q_{XY}$ onto the linear set of distributions having given marginal distributions, if exists, must belong to $\mathcal{P}(Q_{XY})$. Due to this fact, $\mathcal{P}(Q_{XY})$ has emerged, e.g.~in the context of hypercontractivity \cite{kamath15} and multiterminal hypothesis testing \cite[(3)]{polyanskiy2012hypothesis} (for example, the minimizer in \cite[(3)]{polyanskiy2012hypothesis} must lie in $\mathcal{P}(Q_{XY})$). In interactive source coding \cite{ma2012interactive}, $\mathcal{P}(Q_{XY})$ has been defined for discrete distributions, without introducing the preorder relation. In both \cite{ma2012interactive} and the present paper, the appearance of $\mathcal{P}(Q_{XY})$ is due to the conditioning on auxiliary random variables satisfying Markov structures, cf.~\eqref{e_rregion1}-\eqref{e_rregion3}.
\end{rem}

Next, we introduce notions of concave functions and concave envelopes with respect to the marginal distributions, generalizing the discrete case defined in \cite{ma2012interactive}. We refine those definitions using the $XY$-absolute continuity framework to resolve the technicality of defining a conditional distribution from a joint distribution.
\begin{defn}
A functional $\sigma$ on a set $\mathcal{P}$ of distributions is said to be \emph{$X$-concave} if for any $P_{XY}\in\mathcal{P}$, $(P_{XY}^i)_{i=0,1}$ and $\alpha\in[0,1]$ satisfying
\begin{align}
P_{XY}^i&\preceq_X P_{XY},\quad i=0,1;
\\
P_{XY}&=(1-\alpha)P_{XY}^0+\alpha P_{XY}^1,
\end{align}
it holds that
\begin{align}
\sigma(P_{XY})\ge (1-\alpha)\sigma(P_{XY}^0)+\alpha\sigma(P_{XY}^1).
\end{align}
Moreover, $\sigma$ is said to be \emph{$XY$-concave} if it is both $X$-concave and $Y$-concave.
\end{defn}

\begin{defn}\label{defn_env}
Given a functional $\sigma$ on a set $\mathcal{P}$ of distributions, the functional $\sigma'$ is said to be the $X$-concave envelope of $\sigma$, denoted as $\xenv(\sigma)$, if $\sigma'$ is $X$-concave, dominates $\sigma$, and is dominated by any other $X$-concave functional which dominates $\sigma$. The \emph{$XY$-concave envelope}, denoted by $\xyenv(\sigma)$, is defined analogously in terms of $XY$-concavity.
\end{defn}
The existence of the $X$-concave envelope follows from the existence of the conventional concave envelope for a function.
For the existence of $XY$-concave envelope, we can take the $X$-concave envelope and $Y$-concave envelope of the given functional alternatingly,
and the pointwise limit (as the number of taking the marginal concave envelopes tends to infinity) exists by the monotone convergence theorem.
The limit functional satisfies the condition in Definition~\ref{defn_env}.
To summarize, we have
\begin{prop}\label{prop_exist}
Given a functional $\sigma$ on a set $\mathcal{P}$ of distributions,
both $\xenv(\sigma)$ and $\xyenv(\sigma)$ exist.
Moreover, $\xyenv(\sigma)$ equals the pointwise limit of $\sigma_r$, $r\to\infty$ where $\sigma_0:=\sigma$ and
\begin{align}
\sigma_r:=\left\{
\begin{array}{cc}
\env_X(\sigma_{r-1}) & \textrm{$r$ is odd};
\\
\env_Y(\sigma_{r-1}) & \textrm{$r$ is even}.
\end{array}
\right.
\end{align}
\end{prop}

Now return to the examples discussed at the beginning of the subsection.
Fix $Q_{XY}$ and a functional $\sigma$ on $\mathcal{P}_X(Q_{XY})$.
We have
\begin{align}
\xenv\sigma(P_{XY})
=\sup_{P_{U|X}}\mathbb{E}[\sigma(P_{XY|U}(\cdot|U))]
\label{e_concaux}
\end{align}
for any $P_{XY}\in \mathcal{P}_X(Q_{XY})$,
where the supremum is over conditional distribution $P_{U|X}$ for which $\mathcal{U}$ is finite,
and $P_{UXY}=P_{U|X}P_{XY}$.
Thus, characterizations via auxiliary random variables can be reformulated in terms of concave envelopes.
Similarly, rate regions expressed using auxiliaries $U_1,\dots,U_r$ satisfying the Markov structure \eqref{e_markov1}-\eqref{e_markov2}
can be reformulated taking $X$-concave envelopes and $Y$-concave envelopes alternatingly,
which will converge to the $XY$-concave envelope as $r\to\infty$;
this will be explored in detail in the next subsection.

\section{Convex Geometric Characterizations of the Rate Regions}\label{sec_convgeo}
Building on Theorem~\ref{thm_region}, in this section we derive an alternative characterization of the tradeoff between the key rate $R$ and the sum interactive communication rate $R_1+R_2$ (equivalently, a characterization of $\mathcal{S}_r(X,Y)$) in terms of concave envelopes,
which is given in Theorem~\ref{thm_rregion} below.
A similar approach,
which we do not elaborate here,
can be applied to the tradeoff between $R$, $R_1$ and $R_2$ (the region $\mathcal{R}_r(X,Y)$).
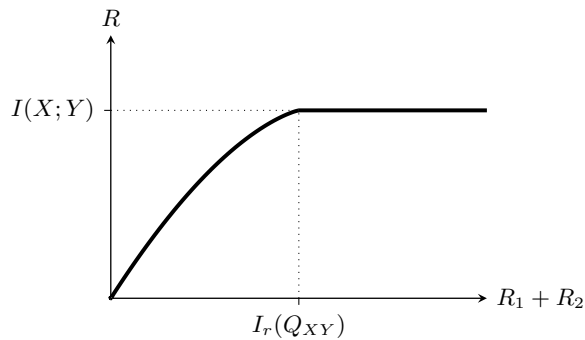
\begin{figure}[ht]
\centering
    \begin{tikzpicture}[scale=5,
      dot/.style={draw,fill=black,circle,minimum size=1mm,inner sep=0pt},arw/.style={->,>=stealth}]
      \draw[arw,line width=0.5pt] (0,0) to (0,0.7) node[above,font=\small] {$R$};
      \draw[arw,line width=0.5pt] (0,0) to (1,0) node[right,font=\small]{$R_1+R_2$};
      \foreach \p in {75}
      {
        \draw[line width=1.5pt] plot[smooth] file {region.table};

       \draw[line width=1.5pt] plot coordinates
       {
       (0.498,0.5)
       (1,0.5)
       };
      }
      \draw[dotted] (0.5,0) to ($(0.5,0.5)$);
      \draw[dotted] (0,0.5) to ($(0.5,0.5)$);
      \draw ($(0.5,0)$) -- ++(0,-1/2pt) node[anchor=north,font=\small] {$I_r(Q_{XY})$};
      \draw (0,0.5) -- ++(-1/2pt,0) node[anchor=east,font=\small] {$I({X;Y})$};
    \end{tikzpicture}
    \caption{\small Achievable region (below the curve) of the key rate $R$ and the sum interactive communication rate $R_1+R_2$. The minimum interaction needed for maximum key rate is denoted as $I_r(Q_{XY})$.
    }
\label{fig_region}
\end{figure}

For any $Q_{XYU^r}$ where the finite-valued auxiliary random variables $U^r$ satisfy the Markov chains \eqref{e_markov1}-\eqref{e_markov2}, denote by $R(Q_{XYU^r})$ the right side of \eqref{e_rregion1}.
Similarly, for the total sum rate, define
\begin{align}
S(Q_{XYU^r})
:=\sum_{i\in\mathcal{O}^r}I(U_i;X|U^{i-1})
+\sum_{i\in\mathcal{E}^r}I(U_i;Y|U^{i-1}).
\end{align}
Observe that
\begin{align}
R(Q_{XYU^r})=& \sum_{i\in\mathcal{O}^r} I(U_i;Y|U^{i-1})
+\sum_{i\in\mathcal{E}^r} I(U_i;X|U^{i-1})
\\
=& I(U_1;Y)
+\sum_{i\in\mathcal{E}_2^r}I(U_i;X|U^{i-1})
+\sum_{i\in\mathcal{O}_2^r}I(U_i;Y|U^{i-1})
\\
=& I(X;Y)-I(X;Y|U_1)
+\sum_{i\in\mathcal{E}_2^r}I(U_i;X|U^{i-1})
+\sum_{i\in\mathcal{O}_2^r}I(U_i;Y|U^{i-1}),
\label{e_usemarkov}
\end{align}
where \eqref{e_usemarkov} used the Markov condition $U_1-X-Y$. Hence by rearranging,
\begin{align}
I(X;Y)-R(Q_{XYU^r})= I(Y;X|U_1)-\sum_{u\in\mathcal{U}_1} R(Q_{YXU_2^r|U_1=u})Q_{U_1}(u)
\label{e18}
\end{align}
Now the key observation is that the right side of \eqref{e18} is similar to the left except that each term is conditioned on $U_1$, the roles of $X$ and $Y$ are switched, and (conditioned on $U_1$) there are $r-1$ auxiliary random variables left. Similarly, we also have
\begin{align}
H(X,Y)-S(Q_{XYU^r})=&H(Y,X|U_1)-
\sum_{u\in\mathcal{U}_1} S(Q_{YXU_2^r|U_1=u})Q_{U_1}(u),
\label{e19}
\end{align}
and similar observations can be made.
In fact, by iterating the process we have
\begin{align}
H(X,Y)-S(Q_{XYU^r})&=H(X,Y|U^r);
\label{e39}
\\
I(X;Y)-R(Q_{XYU^r})&=I(X;Y|U^r).
\label{e40}
\end{align}
In the case of non-discrete $(X,Y)$, we can choose a reference measure and replace the entropy/conditional entropy terms above with relative entropy/conditional relative entropy, at the cost of slightly more cumbersome notation, so there is no loss of generality with this approach.

Given $Q_{XY}$, and $s>0$, define a functional on\footnote{In principle \eqref{e23} can be defined on the set of all distributions on $\mathcal{X}\times \mathcal{Y}$, although for the purpose of computing $\omega_r^s(Q_{XY})$, considering the smaller set of $\mathcal{P}(Q_{XY})$ gives the same result while being more computationally economical.} $\mathcal{P}(Q_{XY})$ by
\begin{align}
\omega^s_{0}(P_{XY}):=sH(\hat{X},\hat{Y})-I(\hat{X};\hat{Y}).
\label{e23}
\end{align}
where $P_{XY}\preceq Q_{XY}$ and $(\hat{X},\hat{Y})\sim P_{XY}$.
For $r\in\{1,2,\dots\}$, define
\begin{align}
\omega_r^s:=\left\{
\begin{array}{cc}
\xenv(\omega_{r-1}^s) & \textrm{$r$ is odd};
\\
\xenv(\omega_{r-1}^s) & \textrm{$r$ is even},
\end{array}
\right.
\label{e29}
\end{align}
and define
\begin{align}
\omega_{\infty}^s
:=\xyenv(\omega_0^s)
\end{align}
which agrees with the pointwise limit of $\omega_r^s$ as $r\to\infty$ (Proposition~\ref{prop_exist}).

The main result of this section unveils the connection between $\omega_r^s(Q_{XY})$ and the achievable rate region:
\begin{thm}\label{thm_rregion}
Fix $Q_{XY}$.
For any $r\in \{0,1,2,\dots,\infty\}$,
\begin{align}
\omega^s_r(Q_{XY})
&=sH(X,Y)-I(X;Y)+\sup_{U^r}\left\{
R(Q_{XYU^r})-sS(Q_{XYU^r})\right\}
\label{e31}
\\
&=sH(X,Y)-I(X;Y)+\sup_{(S,R)\in\mathcal{S}_r(X,Y)}\left\{
R-sS\right\}
\label{e48}
\end{align}
where $U^r$ is finite valued and satisfies \eqref{e_markov1}-\eqref{e_markov2}.
\end{thm}
\begin{proof}
By Theorem~\ref{thm_region},
\eqref{e48} is an
immediate consequences of \eqref{e31}.
From \eqref{e39} and \eqref{e40},
we see that \eqref{e31} is equivalent to the following:
\begin{align}
\omega^s_r(Q_{XY})=
\sup_{U^r}
\{
sH(X,Y|U^r)-I(X;Y|U^r)
\}.
\label{e49}
\end{align}
Note that \eqref{e49} is reduced to the definition of $\omega_r^s$ when $r=0$.
When $r=1$, the validity of \eqref{e49} follows since the supremum of a functional over an auxiliary can be represented in terms of concave envelope,
as explained in the end of Section~\ref{sec_abs}
(see \eqref{e_concaux}).
The $r>1$ case follows by induction: conditioned on $U_r=u$ for any $u\in \mathcal{U}_r$, we can apply the induction hypothesis for $r-1$.
\end{proof}

Because $H(X,Y)$ and $I(X;Y)$ do not involve $U^r$,
Theorem~\ref{thm_rregion} tells us that characterizing the closed convex set $\mathcal{S}_r(X,Y)$ is equivalent to computing $\omega^s_r(Q_{XY})$ for each $s>0$.

The significance of Theorem~\ref{thm_rregion} is that we can sometimes come up with a $XY$-concave function that upper-bounds $\omega^s_{0}$. If the upper-bounding function evaluated at $Q_{XY}$ happens to also be achieved by a known scheme, then we can determine $\omega^s_{\infty}(Q_{XY})$.

\section{Key Bits per Interaction Bit}\label{sec_kbib}
In this section we consider the following fundamental limit for the interaction key generation problem,
which is analogous to the capacity per unit cost \cite{shannon1949}\cite{verdu1990channel} in the context of channel coding.
We show its connection to certain measures of correlations (strong data processing constants), and evaluation them in the binary symmetric and the Gaussian cases.
\begin{defn}\label{defn_KBI}
For $r\in\{1,2,\dots,\infty\}$, $\delta\in[0,1]$, define
$\Gamma_r^{\delta}(X;Y)$
as the maximum real number $\gamma\ge0$ such that there exists a sequence (indexed by $k$) of $r$-round (possibly stochastic) key generation schemes which fulfill the following conditions:
\begin{align}
\liminf_{k\to\infty}\frac{\log|\mathcal{K}|}{\log|\mathcal{W}^r|}
&\ge \gamma;
\label{e_frac}
\\
\lim_{k\to\infty}\log|\mathcal{K}|&=\infty;
\\
\limsup_{k\to\infty}\Delta_k&\le\delta.
\end{align}
where $\Delta_k$ is defined in \eqref{e_mdelta}.
The \emph{key bits per $r$-round interaction bit} is defined as
\begin{align}
\Gamma_r(X;Y):=\inf_{\delta>0}\Gamma_r^{\delta}(X;Y).
\end{align}
The \emph{key bits per interaction bit} (KBIB) is $\Gamma_{\infty}(X;Y)$.
\end{defn}
Note that there is no constraint on the blocklength in Definition~\ref{defn_KBI}.
In particular, the blocklength can grow super-linearly in $\log|\mathcal{W}^r|$, in which case the rates are zero and the fraction in \eqref{e_frac} cannot be written as $\frac{R}{R_1+R_2}$.
Nevertheless, we can shown that
\begin{prop}\label{prop_gamma}
For stationary memoryless sources,
\begin{align}
\Gamma_r(X;Y)=\sup\left\{\frac{R}{R_1+R_2}\colon (R,R_1,R_2)\in\mathcal{R}_r(X,Y)\right\}
\label{e_eta_op}
\end{align}
\end{prop}
\begin{proof}
The $\ge$ part is more clear here.
For the $\le$ part, we need to be a bit careful about the zero rate case ($R=R_1=R_2=0$) mentioned above;
this is handled by Corollary~\ref{cor_converse} in Appendix~\ref{app_conv} which provides an upper-bound on $\frac{|\log|\mathcal{K}|}{|\mathcal{W}^r|}$ that does not depend on the blocklength $n$.
\end{proof}
Next, we provide a compact formula for $\Gamma_{\infty}$ in terms of a ``symmetric strong data processing constant'' which may be viewed as a variant of the strong data processing constant in information theory (cf.~\cite{ahlswede1976spreading}\cite{anantharam2013maximal}\cite{pw_2015}).

\subsection{Symmetric SDPC and $\Gamma_{\infty}$}
In this subsection we discuss the connection between KBIB and certain measures of the correlation between random variables.
To begin with, recall that the \emph{key bits per communication bit} (cf.~\cite{courtade2013outer}\cite{Liu}\footnote{Incidentally, Ahlswede made pioneering contributions to both the strong data processing constant \cite{ahlswede1976spreading} and key generation \cite{ahlswede1993common}, although it appears that he never explicitly reported a connection between the two.}) is the $r=1$ special case of KBIB, and according to \eqref{e_rregion1}-\eqref{e_rregion3}, admits the formula
\begin{align}
\Gamma_1(X;Y)&=\sup_{U\colon U-X-Y}\frac{I(U;Y)}{I(U;X)-I(U;Y)}
\\
&=\frac{s_1^*(X;Y)}{1-s_1^*(X;Y)}.
\label{e_onewayeta}
\end{align}
where the strong data processing constant (cf.~\cite{ahlswede1976spreading}\cite{anantharam2013maximal}\cite{pw_2015}) is defined as
\begin{align}
s_1^*(X;Y)&:=\sup_{U\colon U-X-Y}\frac{I(U;Y)}{I(U;X)}
\label{e30}
\\
&=\sup\left\{\frac{R}{S}\colon (S,R)\in\mathcal{S}_1(X,Y)\right\}
\label{e59}
\end{align}
and we always assume that the suprema are over auxiliary random variables such that the fraction is well defined.
Conventionally, $s_1^*$ is denoted as $s^*$ \cite{ahlswede1976spreading}\cite{anantharam2013maximal},
but in the context of the present paper we will stick with $s_1^*$.
From \eqref{e30}, it is not hard to see that $s_1^*(X;Y)$ has the following equivalent characterization. Recall \eqref{e23} defined a functional on $\mathcal{P}(Q_{XY})$ by
\begin{align}
\omega^s_{0}(P_{XY})=sH(\hat{X})-H(\hat{Y})+(s+1)H(\hat{Y}|\hat{X})
\label{e_32}
\end{align}
where $P_{XY}\preceq Q_{XY}$ and $(\hat{X},\hat{Y})\sim P_{XY}$.

The strong data processing constant admits the following alternative concave-envelope representation,
which will motivate our definition of the key quantity of interest in the interactive case (Definition~\ref{defn_sinfty} ahead).
\begin{prop}\label{prop_s1}
$s_1^*(X;Y)$ is the infimum of $s>0$ such that
\begin{align}
\omega_1^s (Q_{XY})= \omega^s_{0}(Q_{XY})
\end{align}
where we defined $\omega_1^s=\xenv(\omega_0^s)$.
\end{prop}
\begin{proof}
By Theorem~\ref{thm_rregion},
\begin{align}
&\quad\omega_1^s (Q_{XY})- \omega^s_{0}(Q_{XY})
\nonumber
\\
&=\sup_{(S,R)\in\mathcal{S}_1(X,Y)}
\{R-sS\}-
\sup_{(S,R)\in\mathcal{S}_0(X,Y)}
\{R-sS\}
\\
&=\sup_{(S,R)\in\mathcal{S}_1(X,Y)}
\{R-sS\}
\end{align}
and the claim follows from \eqref{e59}.
\end{proof}
If either $X$ or $Y$ are non-discrete, we may choose an arbitrary reference measure and replace the entropies with (the negative of) the relative entropies, so there is no loss of generality with the concave envelope characterization approach. In the discrete case, the concave envelope characterization in Proposition~\ref{prop_s1} is essentially shown by Anantharam et al.~\cite{anantharam2013maximal}, noting that the third term in \eqref{e_32} is linear in $P_X$ for fixed $Q_{Y|X}$. However, by using the framework in Section~\ref{sec_abs}, Proposition~\ref{prop_s1} avoids the challenge of defining a conditional distribution $Q_{Y|X}$ from the possibly non-discrete joint distribution $Q_{XY}$.
Also we remark that,
in contrast to \eqref{e30},
the original definition the strong data processing constant \cite{ahlswede1976spreading} was in terms of the relative entropy:
\begin{align}
s_1^*(X;Y)&:=\sup_{P_X\colon P_X\ll Q_X}\frac{D(P_Y\|Q_Y)}{D(P_X\|Q_X)}
\end{align}
where $P_X\to Q_{Y|X}\to P_Y$.
The equivalence between the relative entropy definition and the mutual information definition was shown in \cite{anantharam2013maximal}.

\begin{rem}
For some simple joint distributions, $s_1^*$ can be computed explicitly.
For the binary symmetric sources (BSS) with error probability $\epsilon$, $s_1^*(X;Y)=(1-2\epsilon)^2$. The scalar Gaussian sources with correlation coefficient $\rho$ has $s_1^*(X;Y)=\rho^2$. For an erasure channel with erasure probability $\epsilon$ and equiprobable input distribution,
we have $s_1^*(X;Y)=1-\epsilon$ and numerical simulation suggests that $s_1^*(Y;X)=\frac{1}{\log\frac{2}{1-\epsilon}}$ for small enough $1-\epsilon$.
Additional examples including the Z-channel or the binary symmetric channel (BSC) with non-equiprobable inputs can be found in \cite{anan_conj}.
\end{rem}

Returning to key bits per interaction bit, we can define a similar notion of data processing constant from a multi-letter expression, or equivalently according to the analysis in Section~\ref{sec_convgeo}, with the following concave envelope characterization:
\begin{defn}\label{defn_sinfty}
Define the \emph{symmetric data processing constant} (SSDPC) $s_{\infty}^*(X;Y)$ as the infimum of $s>0$ such that
\begin{align}
\omega^s_{\infty} (Q_{XY})= \omega^s_{0}(Q_{XY}).
\label{e_34}
\end{align}
where we defined $\omega^s_{\infty}:=\xyenv\omega^s_0$.
\end{defn}
Some basic properties of $s_{\infty}^*$ are as follows:
\begin{prop}
Fix any random variables $X$ and $Y$.
\begin{enumerate}
\item
\begin{align}
s_{\infty}^*(X;Y)=s_{\infty}^*(Y;X);
\label{e_symmetry}
\end{align}
  \item
  \begin{align}
  s_{\infty}^*(X;Y)\in [0,1];
  \label{e_sstar01}
  \end{align}
  \item
  \begin{align}
  s_{\infty}^*(X;Y)\ge s_1^*(X;Y);
  \label{e_s1star}
  \end{align}
  \item If $X'$ and $Y'$ satisfy the Markov chain $X'-X-Y-Y'$, then
      \begin{align}
      s_{\infty}^*(X';Y')\le s_{\infty}^*(X;Y);
      \label{e_sstardp}
      \end{align}
      \item
  \begin{align}
s_{\infty}^*(X;Y)&=\sup\left\{\frac{R}{S}\colon (S,R)\in\mathcal{S}_{\infty}(X,Y)\right\}
\label{e_sstar}
\end{align}
where the expression of $\mathcal{S}_{\infty}(X,Y)$ is given in Theorem~\ref{thm_region}.
\item If $X'$ and $Y'$ are such that
$(X',Y')$ is independent of $(X,Y)$, then
\begin{align}
s_{\infty}^*(X,X';Y,Y')=\max\{
s_{\infty}^*(X;Y),\,
s_{\infty}^*(X';Y')\}.
\label{e_tens}
\end{align}
\end{enumerate}
\end{prop}
\begin{proof}
\eqref{e_symmetry} is immediate from Definition~\ref{defn_sinfty};
\eqref{e_s1star}
follows from Proposition~\ref{prop_s1};
\eqref{e_sstar01} and \eqref{e_sstardp} follow from Theorem~\ref{thm_region} and the data processing property of the mutual information.
\eqref{e_sstar} follows from Theorem~\ref{thm_rregion}.
As with many other problems in information theory (see for example the discussion in \cite{Liu}),
the algebraic steps in proving the converse
for the region $\mathcal{S}_{\infty}(X,Y)$ imply
\begin{align}
\mathcal{S}_{\infty}((X,X'),(Y,Y'))
=\mathcal{S}_{\infty}(X,Y)
+\mathcal{S}_{\infty}(X',Y')
\end{align}
where $+$ denotes the Minkowski sum, which in turn implies
\eqref{e_tens} in view of \eqref{e_sstar}.
\end{proof}
Note that, in contrast to \eqref{e_symmetry},
the conventional strong data processing constant $s_1^*(\cdot)$ is generally not symmetric \cite{anantharam2013maximal}.
The symmetric SDPC is related to the operational quantities by the following:
\begin{thm}\label{thm_gammas}
For any stationary memoryless source,
\begin{align}
\Gamma_{\infty}(X;Y)&=
\frac{s_{\infty}^*(X;Y)}{1-s_{\infty}^*(X;Y)}.
\label{e_gammainfty}
\end{align}
\end{thm}
\begin{proof}
The result follows from \eqref{e_sstar} and
\eqref{e_eta_op}.
\end{proof}

\subsection{Upper-bounding $s_{\infty}^*$ in Terms of the Maximal Correlation}
In this subsection we provide an upper bound on $s_{\infty}^*$, which allows us to determine its value for the binary symmetric sources and the Gaussian sources.

\label{sec_maxcor}
For any $U^r$ satisfying \eqref{e_markov1}-\eqref{e_markov2}
and odd $i\in\{1,2,\dots,r\}$,
the following upper-bound follows from the definition \eqref{e30} of SDPC:
\begin{align}
\frac{I(U_i;Y|U^{i-1})}
{I(U_i;X|U^{i-1})}
&\le \sup_{u^{i-1}}s_1^*(X;Y|U^{i-1}=u^{i-1})
\\
&\le \sup_{P_{XY}\preceq Q_{XY}}s_1^*(\hat{X};\hat{Y})
\label{e41}
\end{align}
where \eqref{e41} follows since it is trivial to check by induction that $Q_{XY|U^{i-1}=u^{i-1}}\preceq Q_{XY}$ for any $u^{i-1}$.
\begin{defn}
For $(X,Y)\sim Q_{XY}$, the
\emph{maximal correlation coefficient} \cite{hirschfeld}\cite{gebelein1941}\cite{renyi}
is defined as
\begin{align}
\rho_{\rm m}^2(X;Y):=\sup_{f,g}\mathbb{E}[f(X)g(Y)]
\end{align}
where the supremum is over measurable real valued functions $f$ and $g$ satisfying $\mathbb{E}[f(X)]=\mathbb{E}[g(Y)]=0$ and $\mathbb{E}[f^2(X)]=\mathbb{E}[g^2(Y)]=1$.
\end{defn}
Ahlswede and G\'{a}cs \cite[Theorem 8]{ahlswede1976spreading} (see also \cite{choi1994equivalence}) proved a useful relation between SDPC and the maximal correlation coefficient, which, in the language of Section~\ref{sec_abs},
is that
\begin{align}\label{e8}
\sup_{P_{XY}\preceq_X Q_{XY}}s_1^*(\hat{X};\hat{Y})=\sup_{P_{XY}\preceq_X Q_{XY}}\rho_{\rm m}^2(\hat{X};\hat{Y})
\end{align}
assuming that $\mathcal{X}$ and $\mathcal{Y}$ are finite.
From Definition~\ref{defn_abs} and \eqref{e8}, we have
\begin{align}
\sup_{P_{XY}\colon P_{XY}\preceq Q_{XY}}s_1^*(\hat{X};\hat{Y})
&=\sup_{T_{XY}\colon T_{XY}\preceq_Y Q_{XY}}\,\sup_{P_{XY}\colon P_{XY}\preceq_X T_{XY}}s_1^*(\hat{X};\hat{Y})
\label{e43}
\\
&=\sup_{T_{XY}\colon T_{XY}
\preceq_Y Q_{XY}}\,\sup_{P_{XY}\colon P_{XY}\preceq_X T_{XY}}\rho_{\rm m}^2(\hat{X};\hat{Y})
\\
&=\sup_{P_{XY}\colon P_{XY}\preceq Q_{XY}}\rho_{\rm m}^2(\hat{X};\hat{Y})
\label{e45}
\end{align}

\begin{thm}
Given $Q_{XY}$,
\begin{align}
s_{\infty}^*(X;Y)&\le \sup_{P_{XY}\preceq Q_{XY}}\rho_{\rm m}^2(\hat{X};\hat{Y}),
\label{e_ub_cr}
\end{align}
where $(\hat{X},\hat{Y})\sim P_{XY}$.
As a consequence, for any stationary memoryless source with per-letter distribution $Q_{XY}$,
\begin{align}
\Gamma_{\infty}(X;Y)&\le \sup_{P_{XY}\preceq Q_{XY}}\frac{\rho_{\rm m}^2(\hat{X};\hat{Y})}{1-\rho_{\rm m}^2(\hat{X};\hat{Y})}.
\label{e_ub}
\end{align}
Moreover, if $P_{XY}=Q_{XY}$ supremizes $\rho_{\rm m}(\hat{X};\hat{Y})$, then both
\eqref{e_ub_cr} and \eqref{e_ub}
hold with equality
and in fact $\Gamma_{\infty}(X,Y)$ can be achieved by the one-way communication protocol (in either way).
\end{thm}
\begin{proof}
For $U^r$ satisfying \eqref{e_markov1}-\eqref{e_markov2},
by \eqref{e41} and \eqref{e45} we have
\begin{align}
\frac{I(U_i;Y|U^{i-1})}
{I(U_i;X|U^{i-1})}
\le
\sup_{P_{XY}\colon P_{XY}\preceq Q_{XY}}\rho_{\rm m}^2(\hat{X};\hat{Y})\label{e_82}
\end{align}
for each odd $i\in \{1,2,\dots,r\}$.
Since the right side of \eqref{e_82} is symmetric in $X$ and $Y$, by the same argument we see that for each even $i\in\{1,2,\dots,r\}$, $\frac{I(U_i;X|U^{i-1})}
{I(U_i;Y|U^{i-1})}$ has can be upper-bounded by the right side of \eqref{e_82} as well.
However by \eqref{e_sstar} and Theorem~\ref{thm_region},
\begin{align}
s_{\infty}^*(X;Y)
=
\sup_{U^r}\frac{\sum_{i\in\mathcal{O}^r}I(U_i;Y|U^{i-1})
+\sum_{i\in\mathcal{E}^r} I(U_i;X|U^{i-1})}{\sum_{i\in\mathcal{O}^r} I(U_i;X|U^{i-1})
+\sum_{i\in\mathcal{E}^r} I(U_i;Y|U^{i-1})}
\end{align}
Therefore \eqref{e_ub_cr} follows from
the fact that $\frac{\sum_i a_i}{\sum_i b_i}\le \sup_i \frac{a_i}{b_i}$ for any nonnegative $(a_i)$ and $(b_i)$ for which the fractions are defined.

The sufficient condition for the equalities can be seen from \eqref{e_onewayeta} and the fact that $\rho_{\rm m}^2(\hat{X};\hat{Y})\le s_1^*(\hat{X};\hat{Y})$ for any $P_{XY}$.
\end{proof}

In general, the maximal correlation coefficient is much easier to compute than the strong data processing constant. Let us use the boldface to denote a matrix corresponding to a discrete joint distribution. e.g.
\begin{align}
\mb{P}_{XY}:=\left[P_{XY}(x,y)\right]_{xy},
\end{align}
with the marginal distributions always expressed as column vectors.
Define
\begin{align}
\mb{A}:=\diag(\mb{P}_X)^{-\frac{1}{2}} \mb{P}_{XY} \diag(\mb{P}_Y)^{-\frac{1}{2}}
\label{e_A}
\end{align}
and let $\mb{M}:=\mb{A}^{\top}\mb{A}$.
Then $\rho_{\rm m}^2(\hat{X};\hat{Y})$ is the second largest eigenvalue value of $\mb{M}$ (cf.~\cite{anantharam2013maximal}). See also \cite{lancaster} for an extension to non-discrete distributions.

Using the calculus of variation, we show next a necessary condition that the discrete distribution $P_{XY}=Q_{XY}$ achieves the supremum in \eqref{e_ub}, whose proof is deferred to Appendix~\ref{app_thm_2}.
\begin{defn}\label{defn_indecomposable}
The \emph{graph} of a discrete distribution distribution $Q_{XY}$ is defined as the bipartite graph whose adjacency matrix is the sign\footnote{That is, a positive entry of ${\bf Q}_{XY}$ yields an entry $1$ in the adjacency matrix, and a zero entry of ${\bf Q}_{XY}$ yields a zero entry in the adjacency matrix.} of ${\bf Q}_{XY}$.
We say $Q_{XY}$ is \emph{indecomposable} \cite{witsenhausen1975sequences} if its graph is connected, and \emph{decomposable} otherwise.
\end{defn}
\begin{thm}\label{thm_2}
Fix $Q_{XY}$ where $|\mathcal{X}|,|\mathcal{Y}|<\infty$.
Let $\mb{u}$ and $\mb{v}$ be the left and right singular vectors of $\mb{A}$ (defined in \eqref{e_A}) corresponding to the second largest singular value of $\bf A$, both normalized to have unit $\ell_2$ norm, so that $\rho_{\sf m}(X;Y)=\mb{u}^{\top}\mb{A}\mb{v}$.
If $P_{XY}=Q_{XY}$ achieves the supremum in \eqref{e_ub}, then
\begin{align}
\mb{u}^{\circ2}&={\bf Q}_{X|Y}\mb{v}^{\circ2};\label{e32}
\\
\mb{v}^{\circ2}&={\bf Q}_{Y|X}\mb{u}^{\circ2},\label{e33}
\end{align}
where $\mb{u}^{\circ2}$ denotes the entry-wise square of $\mb{u}$.
Moreover, if $Q_{XY}$ is indecomposable and both $Q_X$ and $Q_{Y}$ are fully supported, then we further have
\begin{align}
\mb{u}^{\circ2}&={\bf Q}_X;\label{e_full1}
\\
\mb{v}^{\circ2}&={\bf Q}_Y.\label{e_full2}
\end{align}
\end{thm}
\begin{proof}
See Appendix~\ref{app_thm_2}.
\end{proof}
\begin{rem}
Conditions \eqref{e_full1} and \eqref{e_full2} need not be satisfied when $Q_{XY}$ is decomposable (e.g.~consider $X=Y$ binary but not equiprobable under $P_{XY}$).
\end{rem}

Applying Theorem~\ref{thm_2} to BSS, we have the following result.
\begin{thm}\label{lem2}
If $Q_{XY}$ is a BSS with error probability $\epsilon\in[0,1]$, then
\begin{align}
\sup_{P_{XY}\preceq Q_{XY}}\rho_{\rm m}^2(\hat{X};\hat{Y})=(1-2\epsilon)^2.
\label{e_90}
\end{align}
As a consequence, interaction does not increase KBIB for BSS:
\begin{align}
\Gamma_r(X;Y)&=\frac{(1-2\epsilon)^2}{1-(1-2\epsilon)^2}
\end{align}
for any $r\in\{1,2,\dots,\infty\}$.
\end{thm}
\begin{proof}
We may assume without loss of generality that $\epsilon\in(0,1)$. Then by \cite{kamath2012non}, the maximal correlation coefficient is continuous at any $P_{XY}$ with fully supported marginal distribution. It is also elementary to show that $\rho_{\rm m}^2(\hat{X};\hat{Y})$ vanishes as either $P_X$ or $P_Y$ tends to a deterministic distribution. Therefore, the supremum in the definition of $\bar{\rho}_{\rm m}$ is achieved.
We assume for convenience that $Q_{XY}$ is a maximizer, and we will show that $Q_X$ and $Q_Y$ are equiprobable.
By Theorem~\ref{thm_2},
the singular vectors $\mb{u}$ and $\mb{v}$ corresponding to the second singular value $\rho_{\sf m}$ satisfy \eqref{e_full1} and \eqref{e_full2}.
However, the singular vectors associated with the largest singular value $1$ are
$
\mb{Q}_X^{\circ\half}$
and $ \mb{Q}_Y^{\circ\half}
$
(this is true for any $Q_{XY}$).
In the case of $\mathcal{X}=\mathcal{Y}=\{0,1\}$,
the fact that the singular vectors corresponding to the largest and the second largest singular values are orthogonal implies $\mb{u}=\pm(Q_X^{\half}(0),-Q_X^{\half}(1))$
and hence
\begin{align}
Q_X(0)=Q_X(1).
\end{align}
Similarly $Q_Y$ must also be equiprobable.
\end{proof}
\begin{rem}
\cite[Section~6]{ahlswede1976spreading} showed a weaker version of \eqref{e_90}
where the supremization is only over $P_{XY}\preceq_X Q_{XY}$.
see also the related computations in \cite{witsenhausen1975sequences}.
\end{rem}
Theorem~\ref{lem2} may also be proved directly without invoking Theorem~\ref{thm_2}:
\begin{proof}[Second proof of Theorem~\ref{lem2}]
We only need to show that $\rho_{\rm m}^2(\hat{X};\hat{Y})\le(1-2\epsilon)^2$ for any $P_{XY}\in\mathcal{P}(Q_{XY})$. Suppose
\begin{align}
{\bf M}=\left(
          \begin{array}{cc}
            x & \gamma \\
            \beta & y \\
          \end{array}
        \right)
\end{align}
is the matrix such that $\mb{P}_{XY}$ equals the Hadamard product (i.e.\ pointwise product)
\begin{align}
\left(
          \begin{array}{cc}
            \bar{\epsilon} & \epsilon \\
            \epsilon & \bar{\epsilon} \\
          \end{array}
        \right)
        \circ
        \mb{M}
\end{align}
Although ${\bf M}$ is parameterized by four scalars, it only has two degrees of freedom because $P_{XY}\in \mathcal{P}(Q_{XY})$ implies $\mb{M}$ is rank-one (indeed, according to the definition of $\mathcal{P}(\cdot)$ we have $\mb{M}=\mb{f}\mb{g}^{\top}$ for some vectors $\mb{f}$ and $\mb{g}$), and the sum of the coordinates of $\mb{P}_{XY}$ equals 1. In fact, given the sum $s=\beta+\gamma$ and product $p=\beta\gamma$ of the two parameters, we can express the sum and product of $x$ and $y$:
\begin{align}
xy=&p,
\\
x+y=&\frac{1-\epsilon s}{\bar{\epsilon}}.
\end{align}
We know $\rho_{\rm m}(\hat{X};\hat{Y})$ is the second largest singular value of
$\left[\frac{1}{\sqrt{P_{\hat{X}}(x)P_{\hat{Y}}(y)}}
P_{\hat{X}\hat{Y}}(x,y)\right]_{x,y}$. After some 
straightforward algebra, we can express it in terms of $s$ and $p$:
\begin{align}
\rho_{\rm m}^2(\hat{X};\hat{Y})=
\frac{(1-2\epsilon)^2p}
{(1-2\epsilon)^2p+\epsilon(1-2\epsilon)s+\epsilon^2}.\label{e11}
\end{align}
For $P_{\hat{X}\hat{Y}}\in\mathcal{P}(Q_{XY})$, the admissible $s$ and $p$ satisfy
\begin{align}
0\le s\le&\frac{1}{\epsilon},
\label{e_cond1}
\\
0\le p\le& \frac{1}{4}\min\left\{\left(\frac{1-\epsilon s}{\bar{\epsilon}}\right)^2,s^2\right\}.
\label{e_cond2}
\end{align}
Under the conditions \eqref{e_cond1} and \eqref{e_cond2}, it is elementary to show that \eqref{e11} is maximized when $p=\frac{1}{4}$ and $s=1$.
\end{proof}
A celebrated central limit theorem argument by Gross \cite{gross1975logarithmic} showed that Gaussian hypercontractivity can be obtained by BSS hypercontractivity,
which, by the fact that the strong data processing constant can be computed from the hypercontractivity region (see e.g.\ \cite{ahlswede1976spreading}\cite{anantharam2013maximal}), implies that
\begin{align}
s_1^*(X;Y)=\rho^2
\label{e_gross}
\end{align}
for jointly Gaussian $(X,Y)$, where $\rho$ denotes the correlation coefficient.
A similar central limit argument
can be applied to upper-bound the key bits per communication bit \cite[Appendix~D]{Liu}.
Here, a similar argument can again be used to upper-bound the symmetric strong data processing constant and the key bits per interaction bit:
\begin{thm}\label{thm_gaussian}
For any jointly Gaussian distribution $Q_{XY}$ with correlation $\rho$,
\begin{align}
s_{\infty}^*(X;Y)= \rho^2.
\label{e_sgaussian}
\end{align}
As a consequence, interaction does not increase KBIB for the Gaussian source:
\begin{align}
\Gamma_r(X;Y)&=\frac{\rho^2}{1-\rho^2}
\end{align}
for any $r\in\{1,2,\dots,\infty\}$.
\end{thm}
\begin{proof}
The ``$\ge$'' part in \eqref{e_sgaussian} is immediate from \eqref{e_s1star} and \eqref{e_gross}.
Here we only need to prove ``$\le$'' in \eqref{e_sgaussian}.
Suppose this is not the case, then by Theorem~\ref{thm_gammas},
\begin{align}
\Gamma_{\infty}(X;Y)&>\frac{\rho^2}{1-\rho^2}.
\end{align}
Then by Proposition~\ref{prop_gamma}, there exists a sequence of schemes (indexed by blocklength $n$) such that
\begin{align}
\liminf_{n\to\infty}\frac{\log|\mathcal{K}|}
{\log|\mathcal{W}^r|}
&> \frac{\rho^2}{1-\rho^2};
\label{e96}
\\
\lim_{n\to\infty}\log|\mathcal{K}|&=\infty;\label{e_97}
\\
\limsup_{n\to\infty}|Q_{K\hat{K}W^r}-T_{K\hat{K}}Q_{W^r}|&=0
\end{align}
where $T_{K\hat{K}}$ is defined in \eqref{e_tk}.
Now consider a stationary memoryless binary symmetric source $(\tilde{X}_i,\tilde{Y}_i)_{i=1}^{\infty}$ with crossover probability $\epsilon=\frac{1-\rho}{2}$
and $\mathcal{X}=\mathcal{Y}=\{-1,+1\}$.
We can use it to simulate a stationary memoryless Gaussian source with correlation $\rho$.
Indeed, let $L$ be a positive integer, and let $N$ and $\hat{N}$ be i.i.d.\ according to the uniform distribution on $[0,L^{-1/2}]$.
Then according to the multivariate central limit theorem (e.g.\ applying \cite[Theorem~4]{valiant} and then using the smoothness of the Gaussian density),
\begin{align}
\left(
\frac{1}{\sqrt{L}}\sum_{i=1}^L\tilde{X}_i+N,\,
\frac{1}{\sqrt{L}}\sum_{i=1}^L\tilde{Y}_i+\hat{N},\,
\right)
\to (X,Y)
\end{align}
as $L\to\infty$,
where the convergence is in the total variation distance, and $(X,Y)\sim Q_{XY}$.
Thus, by taking sufficiently large ``super symbols'' of size $L$,
we can apply the original encoders and decoders to the simulated source to produce $K$ and $\hat{K}$ from $(\tilde{X}_i,\tilde{Y}_i)_{i=1}^{\infty}$ such that under the true distribution $\tilde{Q}$,
\begin{align}
\limsup_{n\to\infty}|\tilde{Q}_{K\hat{K}W^r}
-T_{K\hat{K}}\tilde{Q}_{W^r}|&=0
\end{align}
while \eqref{e96} and \eqref{e_97} are retained.
This means that
\begin{align}
\Gamma_{\infty}(\tilde{X};\tilde{Y})&>\frac{\rho^2}{1-\rho^2},
\end{align}
in contradiction to Theorem~\ref{lem2}.
\end{proof}
We may also define the \emph{key bits per unit cost} if the communication costs in the two directions differ.
From Theorem~\ref{lem2} and Theorem~\ref{thm_gaussian}
we see that one-round communication is also optimal for achieving this quantity in the case of BSS or Gaussian sources, provided that we have the flexibility to choose the direction of communication which has the lower cost.

\section{Minimum Interaction for Maximum Key Rate}
\label{sec_MIMK}
This section focuses on the minimum interactive communication rate needed for obtaining the maximum
key rate in $r$ rounds, starting from A to B,
which we have denoted by $I_r(Q_{XY})$.
This is the same question considered in \cite{tyagi2013common}, but we derive a new characterization based on concave envelopes and use it to disprove a conjecture in \cite{tyagi2013common}.
Precisely, $I_r(Q_{XY})$ can be defined in the following ways from the rate region or the multi-letter characterization of the rate region:
\begin{align}
I_r(Q_{XY})
&=\inf\{r_1+r_2\colon(I(X;Y),r_1,r_2)\in\mathcal{R}_r(X,Y)\}
\\
&=\inf\{d-I(X;Y)\colon(d,I(X;Y))\in\mathcal{S}_r(X,Y)\}
\\
&=\inf\left\{d\colon\sup_{U^r\colon S(Q_{XYU^r})-I(X;Y)\le d}R(Q_{XYU^r})=I(X;Y)\right\}
\label{e82}
\end{align}
where $S(Q_{XYU^r})$ and $R(Q_{XYU^r})$ were defined in Section~\ref{sec_convgeo}.
We then have the following general concave envelope characterization. Its proof is essentially based on a very simple geometric fact about the supporting hyperplane of a convex set (see Figure~\ref{fig_omega}), which should be applicable to other similar problems as well.
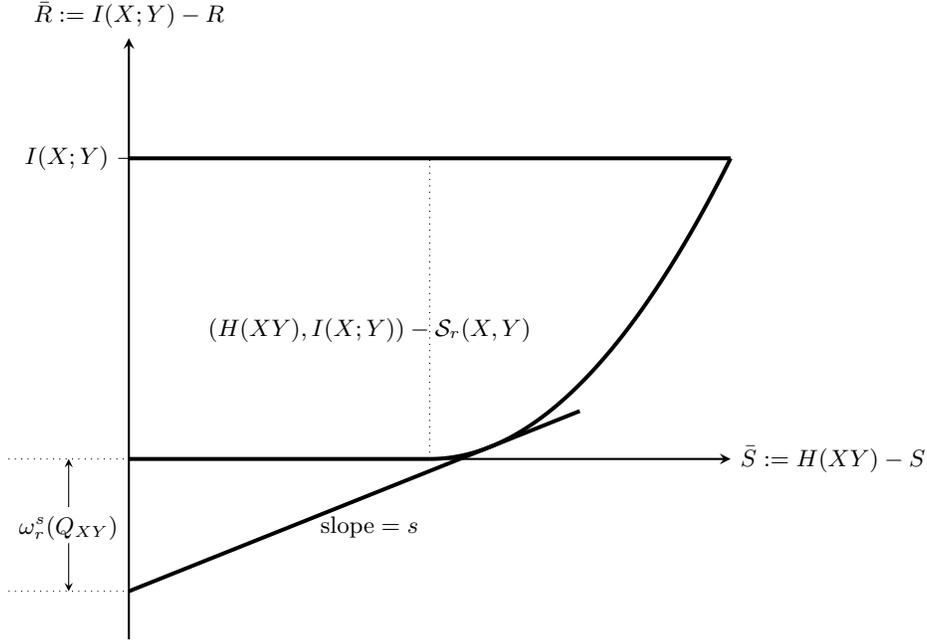
\begin{figure}[ht]
\centering
    \begin{tikzpicture}[scale=8,
      dot/.style={draw,fill=black,circle,minimum size=1mm,inner sep=0pt},arw/.style={->,>=stealth}]
      \draw[arw,line width=0.8pt] (0,-0.3) to (0,0.7) node[above,font=\small] {$\bar{R}:=I(X;Y)-R$};
      \draw[arw,line width=0.8pt] (0,0) to (1,0) node[right,font=\small]{$\bar{S}:=H(XY)-S$};
      \draw[line width=1.5pt] (0,0) to (0.5,0);
      {
       \draw[line width=1.5pt][domain=0.50001:1.001,variable=\r,samples=100]
        plot({\r},{2*(\r-0.5)^2});
    \draw[line width=1.5pt,domain=-0.6:0.15,variable=\r,samples=100]
    plot({\r+0.6},{0.4*\r+0.02});
       \draw[line width=1.5pt] plot coordinates
       {
       (0,0.5)
       (1,0.5)
       };
      }
      \draw[dotted] (0.5,0) to (0.5,0.5);
      \draw[dotted] (-0.2,0) to (0,0);
      \draw[dotted] (-0.2,-0.22) to (0,-0.22);
      \draw[arw] (-0.1,-0.08) to (-0.1,0);
      \draw[arw] (-0.1,-0.14) to (-0.1,-0.22);
      \draw (0,0.5) -- ++(-1/2pt,0) node[anchor=east,font=\small] {$I({X;Y})$};
      \draw (-0.1,-0.08) node[anchor=north,font=\small] {$\omega_r^s(Q_{XY})$};
      \draw (0.4,-0.08) node[anchor=north,font=\small] {${\rm slope}=s$};
            \draw (0.4,0.25) node[anchor=north,font=\small] {$(H(XY),I(X;Y))-\mathcal{S}_r(X,Y)$};
    \end{tikzpicture}
    \caption{\small Geometric illustration of $\omega_r^s(Q_{XY})$.}
\label{fig_omega}
\end{figure}
\begin{thm}\label{thm7}
For a stationary memoryless source with per-letter distribution
$Q_{XY}$,
\begin{align}
I_r(Q_{XY})=H(X|Y)+H(Y|X)-\lim_{s\downarrow0}
\frac{1}{s}\omega_r^s(Q_{XY}).
\label{e80}
\end{align}
where $\omega_r^s$ is as in \eqref{e29}.
\end{thm}
\begin{proof}
From \eqref{e31}, the right side of \eqref{e80} equals
\begin{align}
&\quad\lim_{s\downarrow0}\inf_{U^r}
\left\{S(Q_{XYU^r})-I(X;Y)+\frac{1}{s}[I(X;Y)-R(Q_{XYU^r})]\right\}
\\
&=\sup_{s>0}\inf_{U^r}
\left\{S(Q_{XYU^r})-I(X;Y)+\frac{1}{s}[I(X;Y)-R(Q_{XYU^r})]\right\}.
\label{e_85}
\end{align}
From \eqref{e82}, the infimum in \eqref{e_85} is upper-bounded by $I_r(Q_{XY})$ for any $s$, establishing the $\ge$ part of \eqref{e80}. For the other direction, choose an arbitrary $\epsilon>0$.
Here, note that $\mathcal{S}_r(X,Y)$ is a closed convex subsect of $[0,\infty)\times [0,I(X;Y)]$, and it contains the line $[I(X;Y)+I_r(Q_{XY})]\times\{I(X;Y)\}$.
It is easy to check that $(I(X;Y)+I_r(Q_{XY})-\epsilon,I(X;Y))\notin\mathcal{S}_r(X,Y)$, so by the Hahn-Banach hyperplane separation theorem, there exists an $s>0$ such that for all $U^r$,
\begin{align}
R(Q_{XYU^r})\le s(S(Q_{XYU^r})-I(X;Y)-I_r(Q_{XYU^r})+\epsilon)+I(X;Y).
\end{align}
For such an $s$, the infimum in \eqref{e_85} is lower-bounded by $I_r(Q_{XY})-\epsilon$, as desired.
\end{proof}
Again, for non-discrete distributions we may choose a reference measure and replace the entropy with the relative entropy in the analysis, so a similar result holds, \emph{mutatis mutandis}. However, it should be pointed out that $I_r(Q_{XY})$ is \emph{usually} infinite for non-discrete sources, such as the Gaussian source.

Next we provide an even simpler characterization of the MIMK. Define
\begin{align}
\sigma_0(P_{XY}):=
\left\{
\begin{array}{cc}
  H(\hat{X},\hat{Y}) & I(\hat{X};\hat{Y})=0; \\
  -\infty & \textrm{otherwise},
\end{array}
\right.
\label{e_sigma}
\end{align}
where $(\hat{X},\hat{Y})\sim P_{XY}$.
For $r\in\{1,2,\dots\}$, define
\begin{align}
\sigma_r:=\left\{
\begin{array}{cc}
\env_X(\sigma_{r-1}) & \textrm{$r$ is odd};
\\
\env_Y(\sigma_{r-1}) & \textrm{$r$ is even},
\end{array}
\right.
\end{align}
and define $\sigma_{\infty}$ as the $XY$-concave envelope of $\sigma_0$.
Note that by \eqref{e_85},
\begin{align}
\lim_{s\downarrow0}\frac{1}{s}\omega_r^s(Q_{XY})
=\inf_{s>0}&\sup_{U^r}
\{H(X,Y)-S(Q_{XYU^r})
\nonumber
\\
&\left.-\frac{1}{s}[I(X;Y)-R(Q_{XYU^r})]\right\},
\label{e117}
\end{align}
where $U^r$ is finite and satisfies \eqref{e_markov1}-\eqref{e_markov2}.
In view of \eqref{e39}-\eqref{e40}, we can express $\sigma_r(Q_{XY})$ in a similar form:
\begin{align}
\sigma_r(Q_{XY})
&=\sup_{U^r}
\{H(X,Y)-S(Q_{XYU^r})\,|\,
\nonumber\\
&\quad\quad U^r\colon I(X;Y)-R(Q_{XYU^r})=0\}
\label{e118}
\\
&=\sup_{U^r}\inf_{s>0}
\{H(X,Y)-S(Q_{XYU^r})
\nonumber
\\
&\left.\quad\quad-\frac{1}{s}[I(X;Y)-R(Q_{XYU^r})]\right\}
\end{align}
where in \eqref{e118} notice that $I(X;Y)-R(Q_{XYU^r})=I(X;Y|U^r)$ is always nonnegative.
The next result ensures that in the finite alphabet case, we can indeed switch the order of the supremum and the infimum. As is often the case, compactness (in this case related to the finite alphabet assumption) guarantees such saddle point properties.
\begin{lem}\label{lem4}
Fix $Q_{XY}$ where $|\mathcal{X}|,|\mathcal{Y}|<\infty$. For any $P_{XY}\in \mathcal{P}(Q_{XY})$ and $r\in\{0,1,2,3,\dots,\infty\}$,
  \begin{align}
  \sigma_r(P_{XY})=\lim_{s\downarrow0}\frac{1}{s}\omega_r^s(P_{XY}).
  \label{e_ptwconverge}
  \end{align}
\end{lem}
\begin{proof}
The pointwise convergence \eqref{e_ptwconverge} trivially holds when $r=0$, in view of the definition \eqref{e_sigma}. For other values of $r$, the proof follows by induction, using the fact that $\omega^s_r(P_{XY})$ monotonically decreases in $s$ and Proposition~\ref{prop_conv}. Note that the nonnegativity assumption in Proposition~\ref{prop_conv} because $\frac{1}{s}\omega_r^s(P_{XY})=H(P_{XY})>0$ when either $\hat{X}$ or $\hat{Y}$ is constant.
\end{proof}
\begin{rem}
The proof for the $r<\infty$ case is quite brief; see Proposition~\ref{prop_conv}.\ref{l_conv}) in the Appendix~\ref{app_ptconv}.
The the $r=\infty$ case is slightly trickier,
and we prove that case by first establishing certain semicontinuity of the $XY$-concave functions in Lemma~\ref{lem_semi} in Appendix~\ref{app_cont}.
Remark that for fully supported $Q_{XY}$, the proof of such semicontinuity is simple since Lemma~\ref{lem_semi}.\ref{l1})-\ref{l2}) will become obvious; for general $Q_{XY}$ our proof can be viewed as a natural extension of the idea despite involving additional machineries.
\end{rem}
\begin{rem}
The famous von Neumann min-max theorem \cite{neumann1928} states that if $\mathcal{X}$ and $\mathcal{Y}$ are \emph{compact} convex sets in Euclidean spaces (of possibly different dimensions),
and $f\colon\mathcal{X}\times \mathcal{Y}\to\mathbb{R}$ is a continuous convex-concave function, then
\begin{align}
\min_{x\in\mathcal{X}}\max_{y\in\mathcal{Y}}
f(x,y)
=
\max_{y\in\mathcal{Y}}\min_{x\in\mathcal{X}}
f(x,y).
\end{align}
The quantity inside $\{\}$ in \eqref{e117} can be viewed as a function of $\frac{1}{s}$
and
$
(H(X,Y)-S(Q_{XYU^r}),\,I(X;Y)-R(Q_{XYU^r}))
$,
which is convex-concave (in fact, it is linear-linear).
In the case of finite $\mathcal{X}$ and $\mathcal{Y}$ and $r<\infty$, we can show that the union of
\begin{align}
[0,H(X,Y)-S(Q_{XYU^r})]
\times
[I(X;Y)-R(Q_{XYU^r}),I(X;Y)]
\end{align}
over finite $U^r$ satisfying \eqref{e_markov1}-\eqref{e_markov2} is compact (by the same argument showing that the convex combination of two compact sets in a Euclidean space is compact \cite{rockafellar1970convex}).
The other argument $\frac{1}{s}\in(0,\infty)$ does not have a compact domain, but this is not an essential obstacle when the objective function is linear-linear (for example, one can invoke other min-max theorems such as Lagrange duality \cite{rockafellar1970convex}).
This amounts to an alternative proof of Lemma~\ref{lem4} in the case of $r<\infty$.
\end{rem}

\begin{thm}\label{thm_MIMK}
If $Q_{XY}$ is a distribution on a finite alphabet, then for $r\in\{1,\dots,\infty\}$
\begin{align}
I_r(Q_{XY})=H(Y|X)+H(X|Y)-\sigma_r(Q_{XY}).
\end{align}
\end{thm}
\begin{proof}
Immediate from Lemma~\ref{lem4} and Theorem~\ref{thm7}.
\end{proof}
\begin{rem}
In \cite{tyagi2013common} a quantity called ``interactive common information'' $CI_r(X\wedge Y)$ is defined, and the main result therein is that
\begin{align}
I_r(Q_{XY})=CI_r(X\wedge Y)-I(X;Y).
\end{align}
Hence $CI_r(X\wedge Y)$, $\sigma_r(Q_{XY})$ and $I_r(Q_{XY})$ are closely related.
\end{rem}
\begin{cor}\label{cor_binary}
If $X$ and $Y$ are both binary under $Q_{XY}$, then the necessary and sufficient condition for
\begin{align}
\min\{I_1(Q_{XY}),I_1(Q_{YX})\}=I_{\infty}(Q_{YX})
\label{e_77}
\end{align}
is that either $Q_{Y|X}$ or $Q_{X|Y}$ is a binary symmetric channel ($Q_X$-almost surely or $Q_Y$ almost surely), or $X\perp Y$.
\end{cor}
\begin{rem}
Tyagi \cite{tyagi2013common} introduced a concept called ``interactive common randomness'' and showed its relation to the minimum rate of interactive communication needed to generate the maximum amount of key. Then by drawing an elegant connection to sufficient statistics, Tyagi \cite[Theorem~9]{tyagi2013common} proved Corollary~\ref{cor_binary} in the case of binary symmetric $Q_{XY}$, and conjectured that \eqref{e_77} holds for all binary sources. Here we provide the necessary and sufficient condition for the conjecture to hold following an entirely different approach.
\end{rem}
\begin{proof}[Proof of Corollary~\ref{cor_binary}]
The sufficiency is relatively easy to prove.
If $X$ and $Y$ are independent under $Q_{XY}$,
then \eqref{e_77} trivially holds since both sides equal $0$.
Next, suppose that $Q_{Y|X}$ is a BSC with crossover probability $\epsilon\in[0,1/2)$
and that $Q_X$ is fully supported.
(If $\epsilon=1/2$ or $Q_X$ is not fully supported then $X\perp Y$.
Also the $\epsilon\in(1/2,1]$ case can be argued similarly by symmetry.)
Then
any $P_{XY}\in\mathcal{P}(Q_{XY})$ can be parameterized by $f,g\in [0,1]$ as\footnote{We use the notation $\bar{x}:=1-x$ for $x\in[0,1]$,
and $a*b:=\bar{a}b+a\bar{b}$ denotes the binary convolution of $a,b\in [0,1]$.}
    \begin{align}\label{e_parameterize}
\mb{P}_{XY}=\frac{1}{Z}\left(
                                \begin{array}{cc}
                                  \bar{\epsilon}\bar{f}\bar{g} & \epsilon\bar{f}g \\
                                  \epsilon f\bar{g} & \bar{\epsilon}fg \\
                                \end{array}
                              \right)
\end{align}
with the normalization constant
\begin{align}
Z:=f*g*\bar{\epsilon}.
\end{align}
That is, there exists a one-to-one correspondence from $(f,g)\in[0,1]^2$ to $P_{XY}\in\mathcal{P}(Q_{XY})$.
Let $\pi$ be such a bijection, and $\pi_X(f,g)$ (resp.~$\pi_Y(f,g)$) be the $X$-marginal (resp.~$Y$-marginal) of $\pi(f,g)$.
To avoid cumbersome notations,
for a functional $\sigma$ on $\mathcal{P}(Q_{XY})$
we will sometimes
abbreviate $\sigma(\pi(f,g))$
as $\sigma(f,g)$, but keep in mind that concavity are always with respect to the probability distributions rather than $(f,g)$.

For fixed $\epsilon\in [0,1/2)$, we claim that the functional\footnote{In this paper, $h(\epsilon):=\epsilon\log\frac{1}{\epsilon}+(1-\epsilon)
\log\frac{1}{1-\epsilon}$ denotes the binary entropy function.}
$\ell\colon[0,1]^2\to \mathbb{R}$,
\begin{align}
\ell(f,g)= h(\epsilon)+\frac{c(f-\frac{1}{2})
(g-\frac{1}{2})}{f*\bar{g}*\epsilon},
\end{align}
is $XY$-linear on $\mathcal{P}(Q_{XY})$ and
\begin{align}
\sigma_0(f,g)\le \ell(f,g),
\quad\forall (f,g)\in[0,1]^2,
\label{e_sigmal}
\end{align}
hence $\ell$ also upper bounds $\sigma_{\infty}$.
Indeed, since
\begin{align}
\mb{P}_X
=\frac{1}{Z}
\left(
\begin{array}{c}
  \bar{f}(\epsilon*\bar{g}) \\
  f(\epsilon*g)
\end{array}
\right),
\end{align}
we have
\begin{align}
\ell(f,g)
=h(\epsilon)
+\frac{c(g-\frac{1}{2})}{2}
\left[
\frac{P_X(1)}{\epsilon*g}
-\frac{P_X(0)}{\epsilon*\bar{g}}
\right].
\end{align}
Hence for fixed $g$, we see that $\ell(f,g)$ is linear in $P_X$, meaning that $\ell$ is $X$-concave.
By symmetry, $\ell$ is also $Y$-concave.
Next observe that
there exists a real number $c$ such that
\begin{align}
h\left(\frac{\epsilon g}{\bar{g}*\epsilon}\right)
&\le h(\epsilon)-\frac{c}{2}\cdot\frac{g-\frac{1}{2}}{\bar{g}*\epsilon}
\label{e128}
\\
&=
h(\epsilon)-\frac{c}{4\epsilon}+\frac{c}{4\epsilon\bar{\epsilon}}
\cdot
\frac{\bar{g}\bar{\epsilon}}{\bar{g}*\epsilon}
\label{e129}
\end{align}
for all $g\in[0,1]$.
Indeed, when viewed as
as functions of the binary distribution $\left(\frac{\epsilon g}{\bar{g}*\epsilon},
\frac{\bar{\epsilon} \bar{g}}{\bar{g}*\epsilon}\right)$,
the right side of \eqref{e129} is linear
whereas the left side of \eqref{e128} is concave,
and for any $c\in\mathbb{R}$ both functions have the same evaluation at the equiprobable distribution $(\frac{1}{2},\frac{1}{2})$.
Now, since $\sigma_0(f,g)=-\infty$ if $fg\bar{f}\bar{g}>0$,
where $\sigma_0$ was defined in \eqref{e_sigma},
to establish \eqref{e_sigmal} we only need to consider the case where $fg\bar{f}\bar{g}=0$.
However,
when $f=0$ we have
$\sigma_0(f,g)=h\left(\frac{\epsilon g}{\bar{g}*\epsilon}\right)$,
hence \eqref{e_sigmal} is reduced to
\eqref{e129}.
Similarly, \eqref{e_sigmal} also holds for the cases of $f=1$, $g=0$, or $g=1$ cases,
establishing the claim.

Since $Q_{Y|X}$ is BSC with crossover probability $\epsilon$, the parameter $g=\frac{1}{2}$ for $Q_{XY}$, hence
\begin{align}
h(\epsilon)
&\le \sigma_1(Q_{XY})
\label{e135}
\\
&\le \sigma_{\infty}(Q_{XY})
\\
&\le \ell\left(f,\frac{1}{2}\right)
\\
&=h(\epsilon)
\label{e97}
\end{align}
where \eqref{e135} can be seen from
\begin{align}
\sigma_1(Q_{XY})
\ge
\frac{1}{2}\sigma_0(P^0_{XY})
+\frac{1}{2}\sigma_0(P^1_{XY})
=h(\epsilon)
\end{align}
where
\begin{align}
\mb{P}^0_{XY}
:=\left(
\begin{array}{cc}
  \bar{\epsilon} & \epsilon \\
  0 & 0 \\
\end{array}
\right)
\end{align}
and
\begin{align}
\mb{P}^0_{XY}
=\left(
\begin{array}{cc}
  0 & 0 \\
  \epsilon & \bar{\epsilon} \\
\end{array}
\right).
\end{align}
Thus $\sigma_1(Q_{XY})=\sigma_{\infty}(Q_{XY})$, which, by Theorem~\ref{thm_MIMK}, implies that
\eqref{e_77} holds.

To show necessity, notice first that $|\supp(Q_{XY})|=1$ or $2$ are trivial cases.
Indeed if $X=Y$ or $X\perp Y$ (which includes the case where either $X$ or $Y$ is deterministic) then both sides of \eqref{e_77} are zero.
Then there are only two remaining cases:
\begin{enumerate}
\item $Q_{XY}$ is fully supported.
In this case we assume that there exists $\epsilon\in(0,\tfrac{1}{2})\cup(\tfrac{1}{2},1)$ such that $\mathcal{P}(Q_{XY})$ can again be parameterized as \eqref{e_parameterize}.
(If $\epsilon=1/2$ then $X\perp Y$, and \eqref{e_77} trivially holds.)
Observe that
\begin{align}
\pi_X(f,g)=\left(1-\frac{f(\epsilon*g)}{f*\bar{\epsilon}*g},
\frac{f(\epsilon*g)}{g*\bar{\epsilon}*f}\right)
\end{align}
so it is straightforward to check that the solution of $\lambda\in[0,1]$ to
\begin{align}
\lambda\pi_X(1,g)+\bar{\lambda}\pi_X(0,g)=\pi_X\left(\half,g\right)
\end{align}
is given by
\begin{align}
\lambda=\epsilon*g.
\end{align}
Then by definition,
\begin{align}
\sigma_1\left(\half,g\right)
&= \bar{\lambda}\sigma_0(0,g)+\lambda\sigma_0(1,g)
\label{e145}
\\
&=\bar{\lambda}h\left(\frac{\epsilon g}{\bar{\epsilon}*g}\right)
+\lambda h\left(\frac{\epsilon \bar{g}}{\epsilon *g}\right)
\\
&=-h(\epsilon*g)+h(\epsilon)+h(g)
\\
&\le h(\epsilon)
\label{e104}
\\
&= \sigma_2\left(\half,g\right)
\label{e105}
\end{align}
where
\begin{itemize}
\item \eqref{e145} follows since $\sigma_0(f,g)=-\infty$ when $fg\bar{f}\bar{g}>0$,
which implies that when computing $\xenv(\sigma_0)$ only the boundary points ($f=0$ or $1$) will play a role.
\item To see \eqref{e105},
 notice that
 \begin{align}
 \sigma_2\left(\frac{1}{2},g\right)
 &=\yenv\sigma_1\left(\frac{1}{2},g\right)
 \\
 &\ge \yenv\sigma_0\left(\frac{1}{2},g\right)
 \\
 &\ge \xenv\sigma_0\left(g,\frac{1}{2}\right)
 \label{e152}
 \\
 &= \sigma_1\left(g,\frac{1}{2}\right)
 \label{e153}
 \end{align}
 and
 \begin{align}
 \sigma_2\left(\frac{1}{2},g\right)
 &\le \sigma_{\infty}\left(\frac{1}{2},g\right)
 \\
 &\le \sigma_{\infty}\left(g,\frac{1}{2}\right)
 \label{e155}
 \end{align}
 where \eqref{e152} and \eqref{e155} follow from the symmetry of $\sigma_0$ and $\sigma_{\infty}$, respectively.
 However, in \eqref{e135}-\eqref{e97} we have shown that the right sides of \eqref{e153} and \eqref{e155} are both equal to $h(\epsilon)$.
\end{itemize}
Note that for fixed $g\in(0,1)$, $\sigma_1(\cdot,g)$ is $X$-linear and $\sigma_3(\cdot,g)$ is $X$-concave.
If $g\neq \half$,
then \eqref{e104} becomes an identity,
which implies that $\sigma_3\left(\half,g\right)
\ge \sigma_2\left(\half,g\right)
>\sigma_1\left(\half,g\right)$.
This combined with the fact that
$\sigma_3(\cdot,g)\ge \sigma_1(\cdot,g)$
shows that, in fact,
\begin{align}
\sigma_3(\cdot,g)> \sigma_1(\cdot,g)
\end{align}
except possibly at the endpoints (i.e.~when $f\in\{0,1\}$). In sum, we have shown
\begin{align}
\sigma_3(f,g)> \sigma_1(f,g)
\end{align}
except when $f\in\left\{0,\half,1\right\}$ or $g\in\left\{0,\half,1\right\}$. In other words, if neither $Q_{Y|X}$ nor $Q_{X|Y}$ is a BSC, then
\begin{align}
\sigma_{\infty}(Q_{XY})\ge\sigma_3(Q_{XY})> \sigma_1(Q_{XY}),
\end{align}
and by symmetry, we also have
\begin{align}
\sigma_{\infty}(Q_{XY})\ge\sigma_3(Q_{YX})> \sigma_1(Q_{YX})
\end{align}
which implies that the left side of \eqref{e_77} is strictly larger than the right side.

\item $|\supp(Q_{XY})|=3$. Assume without loss of generality that $Q_{XY}(0,0)=0$. We can parameterize $\mathcal{P}(Q_{XY})$ with $f,g$ via the map
    \begin{align}
    \pi\colon [0,1]^2\setminus\{(0,0)\}&\to \Delta(\mathcal{X}\times \mathcal{Y})
    \\
    (f,g)&\mapsto
    \frac{1}{f+\bar{f}g}
    \left[
    \begin{array}{cc}
      0 & \bar{f}g \\
      f\bar{g} & fg
    \end{array}
    \right].
    \end{align}
    Observe that
    \begin{align}
    \pi_X(f,g)=\left(1-\frac{f}{f+\bar{f}g},
    \frac{f}{f+\bar{f}g}\right),
    \end{align}
    so it is straightforward to check that the solution of $\lambda\in[0,1]$ to
\begin{align}
\lambda\pi_X(1,g)+\bar{\lambda}\pi_X(0,g)
=\pi_X\left(f,g\right)
\end{align}
is given by
\begin{align}
\lambda=\frac{f}{f+\bar{f}g}.
\end{align}
Thus,
\begin{align}
\sigma_1(f,g)&=\lambda\sigma_0(1,g)
+\bar{\lambda}\sigma_0(0,g)
\\
&=\lambda\sigma_0(1,g)
\\
&=\frac{f}{f+\bar{f}g}h(g).
\label{e167}
\end{align}
Next, for fixed $f\neq 0$ put $g=g_f(x):=\frac{xf}{1-\bar{f}{x}}$.
By \eqref{e167},
\begin{align}
\sigma_1(f,g_f(x))=(1-\bar{f}x)h\left(\frac{fx}{1-\bar{f}x}\right).
\end{align}
Then for $x\in(0,1)$,
\begin{align}
\frac{{\rm d}^2}{{\rm d}x^2}\sigma_1(f,g_f(x))
=-\frac{f}{x(1-x)(1-\bar{f}x)}<0.
\label{e169}
\end{align}
Since
\begin{align}
\pi_Y(f,g_f(x))
&=\left(1-\frac{g_f(x)}{f+g_f(x)\bar{f}},\,
\frac{g_f(x)}{f+g_f(x)\bar{f}}\right)
\\
&=(1-x,x),
\end{align}
\eqref{e169} implies that $\sigma_1(f,\cdot)$ is \emph{strictly} $Y$-concave for $f\neq0$.
Now suppose there exist some $(f_0,g_0)$ for which $f_0g_0\bar{f}_0\bar{g}_0\neq0$ such that
\begin{align}
\sigma_1(f_0,g_0)=\sigma_{\infty}(f_0,g_0).
\label{e130}
\end{align}
Since $\sigma_1(\cdot,g_0)$ is linear (caution: in the distribution rather than in $f$), $\sigma_{\infty}(\cdot,g_0)$ is concave, and both functions agree on the endpoints (which can be seen from the fact that taking marginal concave envelopes does not change the boundary values), \eqref{e130} implies that, actually,
\begin{align}
\sigma_1(\cdot,g_0)=\sigma_{\infty}(\cdot,g_0),
\end{align}
and in particular $\sigma_{\infty}(\cdot,g_0)$ is linear. By symmetry, $\sigma_{\infty}(g_0,\cdot)$ is also linear. This is a contradiction since $\sigma_{\infty}(g_0,\cdot)$ and $\sigma_1(g_0,\cdot)$ agree at two points $g=0,g_0$ but the former linear function dominates the latter strictly concave function. Thus \eqref{e130} is impossible, and in particular, we conclude by symmetry that
\begin{align}
\sigma_1(Q_{XY})&<\sigma_{\infty}(Q_{XY});
\\
\sigma_1(Q_{YX})&<\sigma_{\infty}(Q_{YX})=\sigma_{\infty}(Q_{XY}),
\end{align}
as desired.
\end{enumerate}
\end{proof}

\section{Concluding Remarks}
When the communication rate is very low or high enough to achieve the maximum key rate,
Theorem~\ref{lem2} and Corollary~\ref{cor_binary} imply that allowing interaction does not increase
the maximum key rate achievable in the one-way communication scheme for BSS.
These two facts naturally lead to:
\begin{conjecture}\label{conj1}
For a binary symmetric source $(X,Y)$,
\begin{align}
\mathcal{S}_1(X,Y)=\mathcal{S}_{\infty}(X,Y).
\end{align}
\end{conjecture}
For a BSS $(X,Y)$ and under
the A$\to$B unidirectional protocol,
the optimal key-rate--communication-rate is $(I(U;Y),I(U;X)-I(U;Y))$, parameterized by the symmetric Bernoulli auxiliary random variable $U$ satisfying $U-X-Y$. The optimality of such auxiliary random variable can be shown using the concavity of the function $x\mapsto h(\epsilon*h^{-1}(x))$; see also \cite{wyner1973} \cite{witsenhausen1975conditional} or the proof of Proposition~5.3 in \cite{chou2012separation}. What is less obvious is that such a scheme is also optimal among protocols allowing interactions, as Conjecture~\ref{conj1} postulates. If Conjecture~\ref{conj1} holds, then $\mathcal{S}_1(X',Y')=\mathcal{S}_{\infty}(X',Y')$ for any $Q_{X'Y'}$ such that $Q_{Y'|X'}$ is a BSC, and in fact $\mathcal{S}_{\infty}(X',Y')$ will be the intersection between a translation of $\mathcal{S}_r(X,Y)$ and the first quadrant. In Appendix~\ref{app_ineq}, we argue that Conjecture~\ref{conj1} is implied by a conjectured inequality involving four parameters, whose validity has been supported by reasonably extensive numerical computations.

We hope that some of the methods in this paper can be useful beyond the scope of the key generation problem in this paper, we hope some of our methods to become useful in other areas. For example, we have already seen that the
$XY$-absolute continuity framework allows us to define the strong data processing constant directly from a joint distribution without worrying about the technical difficulty of determining the conditional distribution from the joint
in non-standard measurable spaces.
The newly introduced symmetric strong data processing constant (Definition~\ref{defn_sinfty}) has a concave envelope definition very similar to the conventional strong data processing constant, and its significance
is worth exploring
in other contexts as well as its properties.
The techniques used for analyzing the concave envelope characterization, such as expressing the MIMK as a limit as the slope of the supporting line vanishes in Theorem~\ref{thm7} and the minimax result for finite-alphabet distributions
in Lemma~\ref{lem4} (based on fundamental properties of $XY$-concave envelopes in Appendix~\ref{app_cont}-\ref{app_ptconv})
are likely to be useful in the related interactive source coding problem or the broader area of interactive function computation originally studied in \cite{yao1979}, and which has gained recent popularity
among the theoretical computer science community \cite{braverman2012coding}\cite{braverman2012interactive}\cite{braverman2014information}\cite{ganor2014exponential}\cite{braverman2015}.

\section{Acknowledgements}
We cordially thank B.~Ghazi and T.S.~Jayram for pointing out a gap in the proof of
a strong converse bound for CR bits per interaction bit with unlimited rounds of communications
which appeared in an earlier draft (see also \cite[Theorem~5]{ISIT_lcv_interactive}).
The gap was due to the fact that the right side of \eqref{e_ub_cr} does not satisfy a tensorization property (unlike related quantities such as $s_1^*$, $s_{\infty}^*$ and $\rho_{\rm m}$).
This work was supported by the NSF under Grants
CCF-1319304,
CCF- 1116013,
CCF-1319299,
CCF-1350595,
and the Air Force Office of Scientific Research under
Grant FA9550-15-1-0180.

\appendices
\section{Proof of Theorem~\ref{thm_region}:
Achievability}\label{app_region}
We are given a stationary memoryless source $Q_{\sf XY}$ and random transformations $(Q_{{\sf U}_i|{\sf U}^{i-1}{\sf X}})_{i\in\mathcal{O}^r}$ and $(Q_{{\sf U}_i|{\sf U}^{i-1}{\sf Y}})_{i\in\mathcal{E}^r}$.
For convenience, we assume without loss of generality that $r$ is even.
At blocklength $n$,
we consider
\begin{align}
Q_{XY}&:=Q_{\sf XY}^{\otimes n};
\\
Q_{U_i|U^{i-1}X}
&:=Q_{{\sf U}_i|{\sf U}^{i-1}{\sf X}}^{\otimes n},
\quad i=1,3,5,\dots,r-1;
\\
Q_{U_i|U^{i-1}Y}
&:=Q_{{\sf U}_i|{\sf U}^{i-1}{\sf Y}}^{\otimes n},
\quad i=2,4,6,\dots,r.
\end{align}
Notice that we have used the \textsf{sans serif} font to indicate per-letter distributions, and roman font to indicate the $n$-letter extensions.
We first generate a codebook ${\bf u}_1$ of size $M_1M'_1$,
indexed by the multi-index $\mb{w}_1\in\{1,\dots,M_1\}\times\{1,\dots,M'_1\}$,
where each codeword is i.i.d.~according to $Q_{U_1}$.
Then for each codeword $u_1^{\mb{w}_1}$,
generate a codebook ${\bf u}_2^{\mb{w}_1}$ of size $M_2M'_2$ where each codeword is i.i.d.~according to $Q_{U_2|U_1=u_1^{\mb{w}_1}}$.
Likewise, for each $u_2^{\mb{w}_1\mb{w}_2}$ generate the codebook ${\bf u}_3^{\mb{w}_1\mb{w}_2}$, and so on.
The messages sent are $w_1,w_2,w_3\dots$. The receivers have to decode $\hat{w}'_1,\hat{w}'_2,\dots$ which are reconstructions of the unsent part of the indices.
Therefore the reconstructed multi-indices are
\begin{align}
\hat{\bf w}_i&:=(w_i,\hat{w}'_i);
\end{align}
for $i\in\{1,\dots,r\}$.

To see how the likelihood encoder \cite{song} operates, consider the first round first.
Let $\hat{P}_{{\bf W}_1}$ be a distribution under which
${\bf W}_1=(W_1,W'_1)$ is equiprobably distributed on $\{1,\dots,M_1\}\times\{1,\dots,M'_1\}$, and put
\begin{align}
\hat{P}_{{\bf W}_1X}:=\hat{P}_{{\bf W}_1}Q_{X|{\bf W}_1}
\end{align}
where, naturally, $Q_{X|{\bf W}_1}:=Q_{X|U_1}(\cdot|u_1^{{\bf W}_1})$ is defined as the output distribution of the random transform $Q_{X|U}$ when the input is the $U_1$-codeword $u_1^{{\bf W}_1}$.
Then the likelihood encoder is \emph{defined} as the random transformation $\hat{P}_{{\bf W}_1|X}$.
Note that if
\begin{align}
\liminf_{n\to\infty}\frac{1}{n}\log M_1M_1'
> I({\sf U}_1;{\sf X}),
\label{e_rr22}
\end{align}
then
\begin{align}
|\hat{P}_{{\bf W}_1|X}Q_X-\hat{P}_{{\bf W}_1X}|
&=|\hat{P}_{{\bf W}_1|X}Q_X-\hat{P}_{{\bf W}_1|X}\hat{P}_{X}|
\\
&=|Q_X-\hat{P}_{X}|
\label{e_rr24}
\\
&\to 0
\label{e_rregion52}
\end{align}
where \eqref{e_rr24} uses a basic property of the total variation distance;
\eqref{e_rregion52} follows from
the soft-covering lemma (see e.g.\ \cite{song} and the references therein),
and $\to$ indicates convergence in expectation (with respect to the random codebook).
However, in contrast to the real distribution $Q_X\hat{P}_{{\bf W}_1|X}$,
the proxy $\hat{P}_{{\bf W}_1X}$ is much more convenient to analyze, since under $\hat{P}$, i) upon receiving $W_1$, B can reconstruct $W'_1$ by performing channel decoding for the channel $Q_{Y|U_1}$ and the codebook $\left(u_1^{(w_1,w'_1)}\right)_{w'_1}$, and conventional channel coding achievability bounds can be used (see \cite{song} for examples of such analyses in other network information theory problems); ii) $F_1$ and $F_1'$, which correspond to the public message and the key respectively, are exactly independent (see \cite{Liu} for a similar analysis in the context of a one-round key generation model).

Now denote by $\hat{P}_{\hat{\bf W}_1|YW_1}$ the decoder used in i) above.
Let us use $\approx$ to indicate that the total variation distance between two distributions converges to zero in expectation (w.r.t.\ the random codebook).
If
\begin{align}
\sup_{n\to\infty}\frac{1}{n}\log M'< I({\sf U}_1;{\sf Y}),
\label{e_rr26}
\end{align}
then
\begin{align}
\hat{P}_{{\bf W}_1|X}Q_{XY}\hat{P}_{\hat{\bf W}_1|YW_1}
&\approx
\hat{P}_{{\bf W}_1X}Q_{Y|X}\hat{P}_{\hat{\bf W}_1|YW_1}
\label{e_rregion53}
\\
&\approx
\hat{P}_{{\bf W}_1X}Q_{Y|X}\bar{P}_{\hat{\bf W}_1|{\bf W}_1}
\label{e_rregion54}
\\
&=
\hat{P}_{{\bf W}_1}Q_{YX|{\bf W}_1}\bar{P}_{\hat{\bf W}_1|{\bf W}_1}
\label{e_rregion55}
\end{align}
where
\begin{itemize}
  \item \eqref{e_rregion53} is from \eqref{e_rregion52}.
  \item In \eqref{e_rregion54} we defined $\bar{P}_{\hat{\bf W}_1|W_1}$ as the identity transform, and used \eqref{e_rr26} and the channel coding theorem.
Indeed, since the total variation distance is twice of the minimum probability that two random variables are not equal over all couplings,
\begin{align}
\left|\hat{P}_{\hat{\bf W}_1|Y=y,W_1=w_1}
-\bar{P}_{\hat{\bf W}_1|{\bf W}_1={\bf w}_1}\right|
\le 2\mathbb{P}[\hat{\bf W}_1\neq {\bf W}_1|Y=y,{\bf W}_1={\bf w}_1]
\end{align}
for any $y$ and ${\bf w}_1$. Integrating both sides with respect to $\hat{P}_{{\bf W}_1X}Q_{Y|X}$ shows the total variation between the two joint distributions is upper-bounded by twice of the error probability of channel decoding.
\item In \eqref{e_rregion55} we defined, naturally, $Q_{YX|{\bf W}_1}:=Q_{X|{\bf W}_1}Q_{Y|X}$.
\end{itemize}
Notice that the left side of \eqref{e_rregion53} is the true distribution whereas the distribution in \eqref{e_rregion54} is an ideal distribution under which $\hat{\bf W}_1={\bf W}_1$ is equiprobable.
After decoding ${\bf W}_1$, Terminal B uses the likelihood $\hat{P}_{{\bf W}_2|Y{\bf W}_1}$ which is similar to the likelihood encoder $\hat{P}_{{\bf W}_1|X}$ used at A except that now the roles of A, B are switched and everything is conditioned on ${\bf W}_1$.
In short, the actions of B is depicted by
\begin{align}
\hat{P}_{{\bf W}_2\hat{\bf W}_1|YW_1}
=\hat{P}_{{\bf W}_2|Y{\bf W}_1}\hat{P}_{\hat{\bf W}_1|YW_1},
\end{align}
therefore multiplying both sides of \eqref{e_rregion53}-\eqref{e_rregion55} by $\hat{P}_{{\bf W}_2|Y{\bf W}_1}$ and tracing out $\hat{\bf W}_1$, we obtain
\begin{align}
\hat{P}_{{\bf W}_1|X}Q_{XY}\hat{P}_{{\bf W}_2|YW_1}
\approx
\hat{P}_{{\bf W}_1}Q_{YX|{\bf W}_1}\hat{P}_{{\bf W}_2|Y{\bf W}_1}.
\label{e_rregion58}
\end{align}
Note that on the right side of \eqref{e_rregion58} we peeled off a factor of $\hat{P}_{{\bf W}_1}$ while  the remaining part $Q_{YX|{\bf W}_1}\hat{P}_{{\bf W}_2|Y{\bf W}_1}$ is analogous to the product $Q_{YX}\hat{P}_{{\bf W}_1|X}$ we started with at the beginning of the second round, except that we switch the role of A, B and everything is conditioned on ${\bf W}_1$.

We can repeat the steps above to show ``algebraically'' that
in $r$-round interactive communications, the multi-indices ${\bf W}^r$ are close to the equiprobable distribution:
\begin{align}
&\quad\prod_{r=1,3,5,\dots}
\hat{P}_{{\bf W}_i|XW^{i-1}}
\cdot Q_{XY}
\cdot
\prod_{r=2,4,6,\dots}
\hat{P}_{{\bf W}_i|YW^{i-1}}
\nonumber
\\
&=
\prod_{r=3,5,\dots}
\hat{P}_{{\bf W}_i|XW^{i-1}}
\cdot(\hat{P}_{{\bf W}_1|X} Q_{XY}
)\cdot
\prod_{r=2,4,6,\dots}
\hat{P}_{{\bf W}_i|YW^{i-1}}
\\
&\approx
\hat{P}_{{\bf W}_1}
\prod_{r=3,5,\dots}
\hat{P}_{{\bf W}_i|XW^{i-1}}
\cdot( Q_{XY|{\bf W}_1}
\hat{P}_{{\bf W}_2|Y{\bf W}_1}
)\cdot
\prod_{r=4,6,\dots}
\hat{P}_{{\bf W}_i|YW^{i-1}}
\label{e_rregion60}
\\
&\approx
\hat{P}_{{\bf W}_1}
\hat{P}_{{\bf W}_2}
\prod_{r=5,\dots}
\hat{P}_{{\bf W}_i|XW^{i-1}}
\cdot(
\hat{P}_{{\bf W}_3|X{\bf W}_1{\bf W}_2}
Q_{XY|{\bf W}_1{\bf W}_2}
)\cdot
\prod_{r=4,6,\dots}
\hat{P}_{{\bf W}_i|YW^{i-1}}
\label{e_rregion61}
\\
&\approx \dots
\\
&\approx \prod_{i=1}^r\hat{P}_{{\bf W}_i}
\cdot
Q_{XY|{\bf W}^r}
\end{align}
where
\begin{itemize}
  \item \eqref{e_rregion60} has been shown in \eqref{e_rregion58}.
  \item \eqref{e_rregion61} is similar to \eqref{e_rregion60} except that the roles of A, B are switched and everything is conditioned on ${\bf F}_1$. The same arguments work through; indeed the one-shot achievability bounds of the conditional versions of channel coding and resolvability (i.e.~the case with universally known side information) are simple extensions of the unconditional counterparts where the information densities inside the probabilities are replaced by conditional information densities. Also, note that $\hat{P}_{{\bf W}_2|{\bf W}_1={\bf w}_1}=\hat{P}_{{\bf W}_2}$ is the equiprobable distribution which is independent of ${\bf w}_1$.
\end{itemize}
Hence, asymptotically and averaged over the random codebook, the distribution of the indices ${\bf W}_1,\dots,{\bf W}_r$ is close to the equiprobable distribution. Moreover these indices are known to both terminals with high probability.
The achievability proof is thus completed by identifying $W'_1,\dots,W'_r$ with the secret key.
Note that when we apply the induction on the number of rounds, the rate assumptions needed, which are analogous to
\eqref{e_rr22} and \eqref{e_rr26},
are given by
\begin{align}
\liminf_{n\to\infty}\frac{1}{n}\log M_iM'_i
&> I({\sf U}_i;{\sf X}|{\sf U}^{i-1}),
\quad i=1,3,5,\dots;
\\
\liminf_{n\to\infty}\frac{1}{n}\log M_iM'_i
&> I({\sf U}_i;{\sf Y}|{\sf U}^{i-1}),
\quad i=2,4,6,\dots,
\end{align}
and
\begin{align}
\limsup_{n\to\infty}\frac{1}{n}\log M'_i
&< I({\sf U}_i;{\sf Y}|{\sf U}^{i-1}),
\quad i=1,3,5,\dots;
\\
\limsup_{n\to\infty}\frac{1}{n}\log M'_i
&> I({\sf U}_i;{\sf X}|{\sf U}^{i-1}),
\quad i=2,4,6,\dots,
\end{align}
which match \eqref{e_rregion1}-\eqref{e_rregion3}.

\section{Proof of Theorem~\ref{thm_region}:
Converse}\label{app_conv}
In view of the asymptotic equivalence of the performance metrics noted in Remark~\ref{rem_equiv},
the converse part can be seen by taking $\epsilon,\nu\downarrow 0$ in the following bound.
\begin{thm}\label{thm_converse}
Consider a stationary memoryless source with per-letter distribution $Q_{\sf XY}$
a positive integer $r$, and $(R,R_1,R_2)\in (0,\infty)^3$.
Suppose that there exists an $r$-round scheme such that
for some blocklength $n$,
\begin{align}
\sum_{i\in\mathcal{O}^r}\log|\mathcal{W}_i|&\le nR_1;
\\
\sum_{i\in\mathcal{E}^r}\log|\mathcal{W}_i|&\le nR_2;
\\
\log|\mathcal{K}|&\ge nR;
\\
\mathbb{P}[K\neq\hat{K}]&\le\epsilon;
\\
\log|\mathcal{K}|-\min\{H(K|W^r),\,H(\hat{K}|W^r)\}&\le nR\nu,
\end{align}
for some $\epsilon,\nu\in (0,1)$.
Then there exist $(Q_{{\sf U}_i|{\sf U}^{i-1}{\sf X}})_{i\in\mathcal{O}^r}$ and $(Q_{{\sf U}_i|{\sf U}^{i-1}{\sf Y}})_{i\in\mathcal{E}^r}$ such that the triple
\begin{align}
\left(R,\,R_1,\,R_2\right)
+\left(-\frac{\log2}{n}-R\nu -\epsilon R,\,\frac{\log 2}{n}+\epsilon R,\, \frac{\log 2}{n}+\epsilon R\right)
\label{e_tuple}
\end{align}
satisfies \eqref{e_rregion1}-\eqref{e_rregion3}.
\end{thm}
The following auxiliary result has proved useful in the converse proofs of many key generation problems:
\begin{lem}\cite{csiszar1978broadcast}\cite[Lemma~4.1]{ahlswede1993common}\label{lem_rregion1}
For arbitrary random variables $U$, $V$, $X^n$ and $Y^n$,
\begin{align}
I(U;X^n|V)-I(U;Y^n|V)
=\sum_{j=1}^n [I(U;X_j|X^{j-1}Y_{j+1}^nV)-I(U;Y_j|X^{j-1}Y_{j+1}^nV)].
\end{align}
\end{lem}
\begin{proof}[Proof of Theorem~\ref{thm_converse}]
We only consider the case where $r$ is even; the odd case can be proved in a similar fashion.
Let $J$ be equiprobable on $\{1,\dots,n\}$ and independent of $(K,\hat{K},W^r,X^n,Y^n)$.
Define for $i\in\{1,\dots,r\}$,
\begin{align}
\tilde{U}_i&:=(X^{J-1},Y_{J+1}^n,J,W^i);
\end{align}
and set $U_i:=\tilde{U}_i$ for $i\in\{1,\dots,r-1\}$ and $U_r:=(\tilde{U}_r,\hat{K})$.

First, to bound $R_1$, observe that for any $i\in\{1,3,5,\dots,r-1\}$,
\begin{align}
&\quad n[I(U_i;X_J|U^{i-1})-I(U_i;Y_J|U^{i-1})]
\nonumber
\\
&=n[I(W_i;X_J|W^{i-1}X^{J-1}Y_{J+1}^nJ)
-I(W_i;Y_J|W^{i-1}X^{J-1}Y_{J+1}^nJ)]
\label{e_rregion7}
\\
&=I(W_i;X^n|W^{i-1})
-I(W_i;Y^n|W^{i-1})
\label{e_rregion8}
\\
&\le H(W_i)
\label{e_rregion9}
\end{align}
where
\begin{itemize}
\item when $i=1$ we used the independence $X_J\perp(X^{J-1}Y_{J+1}^nJ)$ and
$Y_J\perp(X^{J-1}Y_{J+1}^nJ)$ in the proof of \eqref{e_rregion7}.
\item \eqref{e_rregion8} used Lemma~\ref{lem_rregion1}.
\end{itemize}
Therefore,
\begin{align}
n\sum_{i\in\mathcal{O}^r}[I(U_i;X_J|U^{i-1})-I(U_i;Y_J|U^{i-1})]
&\le
\sum_{i\in\mathcal{O}^r} H(W_i)
\\
&\le nR_1.
\end{align}

To bound $R_2$, note that when $i\in\{2,\dots,r-2\}$ we can perform the computations similar to \eqref{e_rregion7}-\eqref{e_rregion9} to obtain that
\begin{align}
n[I(U_i;Y_J|U^{i-1})-I(U_i;X_J|U^{i-1})]\le H(W_i).
\end{align}
The $i=r$ case needs special care, but still works through:
\begin{align}
n[I(U_i;Y_J|U^{i-1})-I(U_i;X_J|U^{i-1})]
&=I(\hat{K}W_r;Y^n|W^{r-1})-I(\hat{K}W_r;X^n|W^{r-1})
\label{e_rregion21}
\\
&\le H(\hat{K}W_r|W^{r-1})-I(\hat{K}W_r;X^n|W^{r-1})
\\
&=H(\hat{K}W_r|X^nW^{r-1})
\\
&=H(W_r|X^nW^{r-1})+H(\hat{K}|X^nW^r)
\\
&=H(W_r|X^nW^{r-1})+H(\hat{K}|X^nW^rK)
\\
&\le H(W_r)+n\left(\frac{\log 2}{n}+\epsilon R\right) \label{e_rregion16}
\end{align}
where \eqref{e_rregion21} used Lemma~\ref{lem_rregion1}, and \eqref{e_rregion16} used Fano's inequality.
Therefore,
\begin{align}
n\sum_{i\in\mathcal{E}^r}[I(U_i;Y_J|U^{i-1})-I(U_i;X_J|U^{i-1})]
&\le
\sum_{i\in\mathcal{E}^r} H(W_i)+n(n^{-1}+\epsilon R)
\\
&\le nR_2+n\left(\frac{\log2}{n}+\epsilon R\right).
\end{align}

To bound $R$, we first observe a decomposition of a mutual information term associated with the last ($r$-th) round:
\begin{align}
I(U_r;X_J|U^{r-1})
&=I(\hat{K}\tilde{U}_r;X_J|U^{r-1})
\\
&=I(\hat{K}\tilde{U}_r;X_J|\tilde{U}^{r-1})
\\
&=I(\tilde{U}_r;X_J|\tilde{U}^{r-1})
+I(\hat{K};X_J|\tilde{U}^r)\label{e_rregion38}
\end{align}
which allows us to temporarily focus on $\tilde{U}^r$ rather than $U^r$.
For any $i\in\{1,3,5,\dots,r-1\}$, consider the following quantity
\begin{align}
T_i&:=n[I(\tilde{U}_i;Y_J|\tilde{U}^{i-1})
+I(\tilde{U}_{i+1};X_J|\tilde{U}^i)]
\\
&=n[I(\tilde{U}_i;Y_J|\tilde{U}^{i-1})
-I(\tilde{U}_i;X_J|\tilde{U}^{i-1})]
\nonumber
\\
&\quad-n[I(\tilde{U}_i;X_J|\tilde{U}^{i-1})
+I(\tilde{U}_{i+1};X_J|\tilde{U}^i)]
\\
&=I(W_i;Y^n|W^{i-1})
-I(W_i;X^n|W^{i-1})
\nonumber
\\
&\quad+n[I(\tilde{U}_i;X_J|\tilde{U}^{i-1})
+I(\tilde{U}_{i+1};X_J|\tilde{U}^i)]\label{e_rregion26}
\\
&=-I(X^n;W_iW_{i+1}|W^{i-1}Y^n)
+n[I(\tilde{U}_i;X_J|\tilde{U}^{i-1})
+I(\tilde{U}_{i+1};X_J|\tilde{U}^i)]
\label{e_rregion27}
\end{align}
where \eqref{e_rregion26} has been shown in \eqref{e_rregion8}, and
\eqref{e_rregion27} is justified as follows:
\begin{align}
I(W_i;X^n|W^{i-1})
-I(W_i;Y^n|W^{i-1})
&=I(W_i;X^n|W^{i-1}Y^n)
-I(W_i;Y^n|W^{i-1}X^n)
\\
&=I(W_i;X^n|W^{i-1}Y^n)
\\
&=I(W_iW_{i+1};X^n|W^{i-1}Y^n),
\end{align}
where we used the Markov chains $W_i-(W^{i-1},X^n)-Y^n$
 and $W_{i+1}-(W^i,Y^n)-X^n$.
Now from \eqref{e_rregion27},
\begin{align}
\sum_{i\in\mathcal{O}^r}T_i
&=-I(W^r;X^n|Y^n)+nI(\tilde{U}^r;X_J)
\\
&=-I(W^r\hat{K};X^n|Y^n)+nI(\tilde{U}^r;X_J),
\end{align}
where we used $X^n-(Y^n,W^r)-\hat{K}$.
Now we can lower bound
\begin{align}
&\quad
n\sum_{i\in\mathcal{O}^r}[I(U_i;Y_J|U^{i-1})+I(U_{i+1};X_J|U^i)]
\nonumber\\
&=\sum_{i\in\mathcal{O}^r} T_i
+nI(\hat{K};X_J|\tilde{U}^r)
\label{e_r38}
\\
&=-I(W^r\hat{K};X^n|Y^n)
+nI(\hat{K}\tilde{U}^r;X_J)
\\
&=-I(W^r\hat{K};X^n|Y^n)
+nI(\hat{K}W^rX^{J-1}Y_{J+1}^nJ;X_J)
\\
&\ge -I(W^r\hat{K};X^n|Y^n)
+nI(\hat{K}W^rX^{J-1}J;X_J)
\\
&= -I(W^r\hat{K};X^n|Y^n)
+nI(\hat{K}W^r;X_J|X^{J-1}J)\label{e_rregion34}
\\
&= -I(W^r\hat{K};X^n|Y^n)
+I(\hat{K}W^r;X^n)
\\
&=-I(W^r;X^n|Y^n)
+I(\hat{K}W^r;X^n)
\label{e_rr57}
\\
&= H(\hat{K}|W^r)
+[I(W^r;Y^n)-I(W^r;Y^n|X^n)]
-H(\hat{K}|X^nW^r)
\label{e_rr58}
\\
&\ge n\left(R-R\nu-\frac{\log 2}{n}-\epsilon R\right),
\label{e_rr60}
\end{align}
where \eqref{e_r38} uses
\eqref{e_rregion38};
\eqref{e_rregion34} uses the independence $X_J\perp(X^{J-1}J)$;
\eqref{e_rr57} uses $\hat{K}-(Y^n,W^r)-X^n$;
\eqref{e_rr58} follows by algebra;
\eqref{e_rr60} follows from and Fano's inequality and
\begin{align}
&\quad I(W^r;Y^n)-I(W^r;Y^n|X^n)
\nonumber\\
&=
\sum_{i=1}^r [I(W_i;Y^n|W^{i-1})-I(W_i;Y^n|X^nW^{i-1})]
\\
&\ge 0
\end{align}
where the nonnegativity of each summand follows since for each $i$, either $W_i-(X^n,W^{i-1})-Y^n$ or $W_i-(Y^n,W^{i-1})-X^n$ holds.

The proof is finished by identifying
\begin{align}
(X,Y,U^k,V^k)&\leftarrow (X_J,Y_J,U^k,V^k)
\end{align}
in the single-letter formula.
\end{proof}
The following bound on $\Gamma_{\infty}^{\delta}(X;Y)$ (Definition~\ref{defn_KBI}) can be obtained from Theorem~\ref{thm_converse}.
Note that the blocklength $n$ does not appear in the bound.
(For a similar result in the context of channel coding with costs, see \cite{verdu1990channel}).
\begin{cor}\label{cor_converse}
Consider a stationary memoryless source with per-letter distribution $Q_{\sf XY}$.
For any $r$-round scheme (with arbitrary blocklength $n$, which does not affect the bound),
\begin{align}
\frac{\log|\mathcal{K}|}
{\log|\mathcal{W}^r|}
\le
\frac{s}{1-s}
\left(
1-\frac{7-5s}{1-s}\delta
-\left(
2\delta\log\frac{1}{2\delta}
+\frac{1+s}{1-s}\log 2
\right)
\frac{1}{\log |\mathcal{K}|}
\right)^{-1}
\end{align}
where the error $\delta\in (0,1)$ is defined as the right side of \eqref{e_mdelta}, and
\begin{align}
s
&:=\sup_{U_1,\dots,U_r}\frac{\sum_{i\in\mathcal{O}^r} I(U_i;Y|U^{i-1})
+\sum_{i\in\mathcal{E}^r} I(U_i;X|U^{i-1})}
{\sum_{i\in\mathcal{O}^r} I(U_i;X|U^{i-1})
+\sum_{i\in\mathcal{E}^r} I(U_i;Y|U^{i-1})}
\\
&=\sup_{(S,R)\in\mathcal{S}_r(X,Y)}\frac{R}{S}
\end{align}
where $U_1,\dots,U_r$ satisfy the same constraints as in \eqref{e_rregion1}-\eqref{e_rregion3}.
\end{cor}
\begin{proof}
As shown in Remark~\ref{rem_equiv}, we have
\begin{align}
\mathbb{P}[K\neq \hat{K}]&\le \delta;
\\
\log|\mathcal{K}|-\min\{H(K|W^r),\,H(\hat{K}|W^r)\}
&\le
2\delta\log \frac{|\mathcal{K}|}{2\delta}
+4\delta\log |\mathcal{K}|.
\end{align}
The result then follows by rearrangements of \eqref{e_mdelta}.
\end{proof}

\section{Semicontinuity of $XY$-concave Functions}
\label{app_cont}
Recall that a concave function on a simplex which is lower bounded (or more or less equivalently, nonnegative) on the vertices is necessarily lower semicontinuous (cf.~\cite[Theorem~10.2]{rockafellar1970convex}). For $XY$-concave functions, we prove a similar basic result, which will be used in the proof of Proposition~\ref{prop_conv}.
\begin{lem}\label{lem_semi}
Given a distribution $Q_{XY}$ where $Q_X$ and $Q_Y$ are fully supported
and $\mathcal{X}=\{1,\dots,m\}$, $\mathcal{Y}=\{1,\dots,n\}$.
\begin{enumerate}
\item\label{l1} If $Q_{XY}$ is indecomposable (see Definition~\ref{defn_indecomposable}), then
      for any $P_{XY}\in\mathcal{P}(Q_{XY})$, there exists a unique $(\mb{f},\mb{g})\in\mathbb{R}^m\times\mathbb{R}^n$ such that
    \begin{align}
    \mb{f}\mb{g}^{\top}\circ\mb{Q}_{XY}=\mb{P}_{XY}
    \end{align}
    where $\circ$ denotes the pointwise product of matrices,
    and the first coordinate $f_1=1$.
\item \label{l2} Fix an $S_{XY}\in\mathcal{P}(Q_{XY})$. For any $\delta\in(0,1)$, there exists an $\epsilon>0$ such that for any (possibly unnormalized) $\mu_{XY}\preceq S_{XY}$ satisfying $|\mu_{XY}-S_{XY}|\le \epsilon$, we can find $T_{XY}$ satisfying
\begin{align}
S_{XY}&\preceq_X T_{XY};\label{e_10}
\\
T_{XY}&\preceq_Y \mu_{XY};\label{e_11}
\\
|S-T|&\le\delta;\label{e12}
\\
|T-\mu|&\le\delta.\label{e13}
\end{align}
Note that \eqref{e_10} and \eqref{e_11} imply that, actually, $S\sim T\sim\mu$.
\item \label{l3} A $XY$-concave function on $\mathcal{P}(Q_{XY})$ which is nonnegative for $P_{XY}\in\mathcal{P}(Q_{XY})$ under which either $X$ or $Y$ is deterministic is necessarily lower semicontinuous.
\end{enumerate}
\end{lem}
\begin{proof}
\begin{enumerate}
  \item Since the graph of $Q_{XY}$ is connected, we can start from $f_1$ and visit all vertices of the bipartite graph to see that all the coordinates of $\mb{f}$ and $\mb{g}$ are uniquely determined.
  \item It is without loss of generality to only prove the case of $S_{XY}=Q_{XY}$. Suppose the graph of $Q_{XY}$ has $k$ connected components, and assume without loss of generality that $1,\dots,k$ are $X$-vertices belonging to different connected components. Consider
  \begin{align}
  \pi\colon \mathbb{R}^{m-k}\times \mathbb{R}^n&\to \mathbb{R}^{mn}
  \\
  (\bar{\mb{f}},\mb{g})&\mapsto \mb{f}\,\mb{g}^{\top}\circ\mb{Q}_{XY}
  \end{align}
  where $\mb{f}$ is an $m$-vector whose first $k$ coordinates are $1$ and last $(m-k)$-coordinates are $\bar{\mb{f}}$. Denote by $\mb{e}_l$ the $l$-vector ($l\ge1$) whose coordinates are all $1$.
  Then $\pi$ is an embedding from a neighborhood of $\mb{e}_{m+n-k}$ to $\mathbb{R}^{mn}$ (cf.~\cite{boothby2003introduction}), because it is standard to check that the rank of the differential of $\pi$ at $\mb{e}_{m+n-k}$ is $m-k+n$ (full rank), where the calculation is essentially reduced to the case of an indecomposable distribution and the result in part \ref{l1}) can be used. Thus there is an open neighborhood $\mathcal{O}$ of $\mb{e}_{m+n-k}$ homeomorphic to its image under $\pi$, and in particular $\pi$ has a continuous inverse on ${\pi(\mathcal{O})}$.
  Consequently, there exists $\epsilon\in(0,1)$ such that if $\mu_{XY}\preceq Q_{XY}$ and $|\mu-Q|\le \epsilon$, then $\mu\in\pi(\mathcal{O})$ and, with $\bf M$ defined as the matrix of $\mu_{XY}$, $(\bar{\bf f},\bf{g})=(\pi|_{\mathcal{O}})^{-1}(P)$ satisfies
  \begin{align}
  \|\mb{f}-\mb{e}_m\|_{\infty}< \delta/2;
  \label{e16}
  \\
  \|\mb{g}-\mb{e}_n\|_{\infty}< \delta/4,
  \end{align}
  where ${\bf f}=({\bf e}_k,\bar{\bf f})$ as before. Put $\nu_{XY}(x,y)=f(x)Q_{XY}(x,y)$, and observe that \eqref{e_10}-\eqref{e_11} are satisfied because $\mb{f}$ and $\mb{g}$ have strictly positive coordinates. Also,
  \begin{align}
  |Q-\nu|&\le \|\mb{e}_m-\mb{f}\|_{\infty}\sum_{x,y}Q_{XY}(x,y)
  \\
  &\le\delta/2;
  \\
  |\nu-\mu|&\le \|\mb{e}_n-\mb{g}\|_{\infty}\sum_{x,y}f(x)Q_{XY}(x,y)
  \\
  &\le\|\mb{e}_n-\mb{g}\|_{\infty}\sum_{x,y}2Q_{XY}(x,y)
  \\
  &\le\delta/2.
  \end{align}
  But from \eqref{e16}, $1-\delta/2<|\nu|<1+\delta/2$, so the probability distribution $T:=\frac{1}{|\nu|}\nu$ satisfies
  \begin{align}
  |\nu-T|< \delta/2.
  \end{align}
  Then \eqref{e12}-\eqref{e13} holds by the triangle inequality.

  \item Consider an $S_{XY}\in\mathcal{P}(Q_{XY})$. Denote by $a>0$ the minimum nonzero entry of $S_{XY}$, and assume without loss of generality that $S_X$ and $S_Y$ are supported on $\{1,\dots,m_1\}$ and $\{1,\dots,n_1\}$, respectively. For any $\delta\in(0,a/4)$, find $\epsilon>0$ as in \ref{l2}).
      For any $R\in\mathcal{P}(Q_{XY})$ satisfying
      \begin{align}
      |R-S|\le \epsilon,
      \label{e24}
      \end{align}
      define, for $x\in\{1,\dots,m\}$ and $y\in\{1,\dots,n\}$,
      \begin{align}
      \mu(x,y):=R(x,y)1\{x\le m_1,y\le n_1\}.
      \end{align}
      Invoke \ref{l2}) and find $T$ satisfying \eqref{e_10}-\eqref{e13}. We have
  \begin{align}
T\ge \frac{a-\delta}{a}S,
\end{align}
so that
\begin{align}
T=\frac{a-\delta}{a}S+\sum_{x\in\mathcal{X}}
\lambda_x D_x
\end{align}
where each $D_x\preceq_X T$ is a distribution under which $X$ is deterministic, and $\sum_x{\lambda_x}=\frac{\delta}{a}$. Denote by $\sigma$ the $XY$-concave function in question.
By its marginal concavity,
\begin{align}
\sigma(T)&\ge\left(1-\frac{\delta}{a}\right)\sigma(S)
+\sum_{x\in\mathcal{X}}
\lambda_x \sigma(D_x)
\\
&\ge\left(1-\frac{\delta}{a}\right)\sigma(S).
\label{e_29}
\end{align}
Since the minimum nonzero entry in $T$ is at least $a-\delta>a/2$,
we have
\begin{align}
\tilde{R}\ge\mu\ge \frac{a-2\delta}{a}T,
\end{align}
where $\tilde{R}:=\frac{1}{|\mu|}\mu=R_{XY|X\le m_1,Y\le n_1}$,
so a similar argument also shows that
\begin{align}
\sigma\left(\tilde{R}\right)\ge\left(1-\frac{2\delta}{a}\right)\sigma(T).
\label{e_31}
\end{align}
Moreover, consider $\tilde{R}^1:=R_{XY|Y\le n_1}$.
Since $1-\epsilon\le|\mu|\le1$ by \eqref{e24},
we have
\begin{align}
\tilde{R}&\preceq_X \tilde{R}^1;
\\
\tilde{R}^1&\preceq_Y R;
\\
(1-\epsilon)\tilde{R}&\le \tilde{R}^1;
\\
(1-\epsilon)\tilde{R}^1&\le R,
\end{align}
so applying the similar argument again,
\begin{align}
\sigma(\tilde{R}^1)&\ge (1-\epsilon)\sigma(\tilde{R});
\label{e_36}
\\
\sigma(R)&\ge (1-\epsilon)\sigma(\tilde{R}^1).
\label{e_37}
\end{align}
Assembling \eqref{e_29}, \eqref{e_31}, \eqref{e_36}, \eqref{e_37} and noting that $\delta$ and $\epsilon$ can be chosen to be arbitrarily small, we must have
\begin{align}
\liminf_{R\to S}\sigma(R)\ge \sigma(S).
\end{align}
\end{enumerate}
\end{proof}

\section{Pointwise Convergence of Marginally Concave Envelopes}
\label{app_ptconv}
The following result forms the basis of the proof of Lemma~\ref{lem4}, the assumptions of which resemble Dini's theorem in real analysis.
\begin{prop}\label{prop_conv}~
\begin{enumerate}
\item\label{l_conv}
Suppose $(f_s)_{s\in(0,\infty)}$ is a family of continuous functions on a simplex $\Delta$, where $f_s(x)$ is nondecreasing in $s$ for any $x\in\Delta$. Define $f(x):=\lim_{s\downarrow 0}f_s(x)$. If $\env f$ is nowhere $-\infty$, then
\begin{align}
\env f(x)=\lim_{s\downarrow0}\env f_s(x)
\end{align}
for any $x\in\Delta$.
\item \label{l_conv2}
Consider a $Q_{XY}$ on a finite alphabet with fully supported $Q_X$ and $Q_Y$.
Suppose $(f_s)_{s\in(0,\infty)}$ is a family of continuous functions on a $\mathcal{P}(Q_{XY})$, where $f_s(P_{XY})$ is nondecreasing in $s$ for any $P_{XY}\in\mathcal{P}(Q_{XY})$, and $f_s$ is nonnegative when either $X$ or $Y$ is deterministic. Define $f(P_{XY}):=\lim_{s\downarrow 0}f_s(P_{XY})$ for each $P_{XY}\in\mathcal{P}(Q_{XY})$. Then
\begin{align}
\xyenv f(Q_{XY})=\lim_{s\downarrow0}\xyenv f_s(Q_{XY}).
\end{align}
\end{enumerate}
\end{prop}
\begin{rem}
There are simple counterexamples to show that, in general, the limit and the concave envelope cannot be switched if a sequence of continuous functions is only assumed to converge pointwise to a certain continuous function. Moreover if the functions are decreasing but not necessarily continuous, the switching can also fail. Therefore both the monotonicity of $\omega^s_r(P_{XY})$ in $s$ and the continuity in $P_{XY}$ play an essential role in the proof of Lemma~\ref{lem4}.
\end{rem}
\begin{proof}
For \ref{l_conv}), the $\env f(x)\le\lim_{s\downarrow0}\env f_s(x)$ part is trivial. For the opposite direction, notice that the following statements are equivalent:
\begin{align}
&\env f+\epsilon>f_s \quad\textrm{for some $s>0$}
\label{e_65}
\\
\Longleftrightarrow&\,
\{x\colon\env f(x)+\epsilon-f_s(x)\le 0\}=\emptyset \quad\textrm{for some $s>0$}
\label{e66}
\\
\Longleftrightarrow&\,
\bigcap_{s>0}\{x\colon\env f(x)+\epsilon-f_s(x)\le 0\}=\emptyset
\label{e67}
\\
\Longleftrightarrow&\,
\sup_{s>0}(\env f+\epsilon-f_s)>0
\\
\Longleftrightarrow&\,
\env f+\epsilon>\inf f_s=f
\label{e_69}
\end{align}
where \eqref{e67} is the main step which follows from Cantor's intersection theorem. More precisely, notice that a concave function on a simplex is lower semicontinuous \cite[Theorem~10.2]{rockafellar1970convex}, so $\env f+\epsilon-f_s$ is lower semicontinuous, and the set in \eqref{e66} is closed in $\Delta$, hence compact. Then \eqref{e67} follows because a decreasing nested sequence of non-empty compact subsets of the Euclidean space has nonempty intersection (Cantor's intersection theorem).
Since \eqref{e_69} holds for all $\epsilon$, we have from \eqref{e_65} and the concavity of $\env f+\epsilon$ that
\begin{align}
\env f+\epsilon\ge \env f_s
\end{align}
for some $s$. Therefore $\env f(x)\ge\lim_{s\downarrow0}\env f_s(x)$
must hold because $\epsilon$ is arbitrary.

The proof of \ref{l_conv2}) is similar. We need the semicontinuity of the $XY$-concave function proved in Lemma~\ref{lem_semi}.\ref{l3}).
\end{proof}

\section{Proof of Theorem~\ref{thm_2}}\label{app_thm_2}
Consider a small perturbation, parameterized by ${\rm d}\mb{f}$, where $\mb{f}$ is a vector of dimension $|\mathcal{X}|$ so that
\begin{align}
{\rm d}\mb{P}_{XY}=\diag({\rm d}\mb{f}) \mb{P}_{XY}.
\label{e_dpxy}
\end{align}
To ensure that the total probability is preserved under the perturbation,  ${\rm d}\mb{f}$ must satisfy
\begin{align}
\mb{P}_{X}^{\top} {\rm d}\mb{f}=0.
\end{align}
Recall that $\mb{A}:=\diag(\mb{P}_X)^{-\frac{1}{2}}\mb{P}_{XY}
\diag(\mb{P}_Y)^{-\frac{1}{2}}$
and $\mb{M}:=\mb{A}^{\top}\mb{A}$,
and $\mb{u}$ and $\mb{v}$ are the left and the right singular vectors of $\mb{A}$ corresponding to the second largest singular value of $\mb{A}$, which is $\rho_{\sf m}(X;Y)$.
In other words, suppose $\mb{A}$ has the following singular value decomposition
\begin{align}
\mb{A}=
\mb{P}\mb{\Lambda}\mb{Q}^{-1}
\end{align}
where $\mb{P}$ and $\mb{Q}$ are orthogonal matrices and $\mb{\Lambda}$ is a nonnegative diagonal matrix whose second diagonal value is $\rho_{\sf m}(X;Y)$.
Then $\mb{u}$ and $\mb{v}$ are the second columns of $\mb{P}$ and $\mb{Q}$ respectively,
hence
\begin{align}
\rho_{\sf m}(X;Y)&=\mb{u}^{\top} \mb{A}\mb{v};
\\
\mb{A}\mb{v}&=\rho_{\sf m}\mb{u};
\label{e299}
\\
\mb{u}^{\top}\mb{A}&=\rho_{\sf m}\mb{v}^{\top}.
\end{align}
By the definitions of $\bf u$ and $\bf v$, we have
\begin{align}
{\rm d}\rho_{\sf m}(X;Y)
&=\frac{1}{2\rho_{\sf m}(X;Y)}{\rm d}\rho^{2}_{\sf m}(X;Y)
\\
&=\frac{1}{2\rho_{\sf m}(X;Y)}\mb{v}^{\top}{\rm d}\mb{M}\mb{v}
\label{e298}
\\
&=\frac{1}{2\rho_{\sf m}(X;Y)}\mb{v}^{\top}
({\rm d}\mb{A}^{\top}\mb{A}+\mb{A}^{\top}{\rm d}\mb{A})\mb{v}
\\
&=\mb{u}^{\top}{\rm d}\mb{A}\mb{v},
\label{e300}
\end{align}
where
\begin{itemize}
\item \eqref{e298} follows since
assuming that $\mb{v}$ maintains the unit $\ell_2$ norm under the perturbation,
\begin{align}
{\rm d}(\mb{v}^{\top}\mb{M}\mb{v})
&=2\mb{v}^{\top}\mb{M}{\rm d}\mb{v}
+\mb{v}^{\top}{\rm d}\mb{M}\mb{v}
\\
&=2\rho_{\sf m}(X;Y)\mb{v}^{\top}{\rm d}\mb{v}
+\mb{v}^{\top}{\rm d}\mb{M}\mb{v}
\\
&=\rho_{\sf m}(X;Y){\rm d}(\mb{v}^{\top}\mb{v})
+\mb{v}^{\top}{\rm d}\mb{M}\mb{v}
\\
&=\mb{v}^{\top}{\rm d}\mb{M}\mb{v}.
\end{align}
\item \eqref{e300} follows from \eqref{e299}.
\end{itemize}
But from \eqref{e_dpxy},
\begin{align}
{\rm d}{\bf P}_{X}&:={\rm d}{\bf f}\circ {\bf P}_{X};
\\
{\rm d}{\bf P}_{Y}&:={\rm d}{\bf f}^{\top}{\bf P}_{XY}.
\end{align}
Let $\mb{P}_{Y|X}$ be the $|\mathcal{Y}|\times |\mathcal{X}|$ matrix with entries being the conditional probabilities, which is invariant under the perturbation ${\rm d}\mb{f}$.
Define $\mb{P}_{X|Y}$ similarly as an $|\mathcal{X}|\times |\mathcal{Y}|$ matrix.
We have
\begin{align}
&\quad{\rm d}\left(\diag(\mb{P}_X)^{\frac{1}{2}}\mb{P}_{Y|X}^{\top}
\diag(\mb{P}_Y)^{-\frac{1}{2}}\right)
\nonumber\\
&=
\frac{1}{2}
\diag({\rm d}\mb{P}_X)
\diag(\mb{P}_X)^{-\frac{1}{2}}
\mb{P}_{Y|X}
\diag(\mb{P}_Y)^{-\half}
\nonumber\\
&\quad
-\half \diag(\mb{P}_X)^{\half}
\mb{P}_{Y|X}^{\top}\diag(\mb{P}_Y)^{-\frac{3}{2}}
\diag({\rm d}\mb{P}_Y)
\\
&=\half \diag({\rm d}\mb{f})\mb{A}
-\frac{\mb{A}}{2}\diag({\rm d}\mb{f}^{\top}
\mb{P}_{X|Y}).
\end{align}
Hence
\begin{align}
\mb{u}^{\top}{\rm d}\mb{A}\mb{v}
&=\frac{{\rm d}{\bf f}^{\top}}{2}\diag(\mb{u}){\bf Av}
-\frac{1}{2}\mb{u}^{\top}\mb{A}\diag({\rm d}{\bf f}^{\top}\mb{P}_{X|Y})\mb{v}
\\
&=\frac{\rho_{\sf m}(X;Y)}{2}(\mb{u}^{\circ2})^{\top} {\rm d}{\bf f}
-\frac{\rho_{\sf m}(X;Y)}{2}(\mb{v}^{\circ2})^{\top}(\mb{P}_{X|Y})^{\top}{\rm d}{\bf f}.
\end{align}
This implies that we must have
\begin{align}
\mb{u}^{\circ2}-\mb{P}_{X|Y}\mb{v}^{\circ2}=a\mb{P}_{X}
\end{align}
for some real number $a$. Summing up the entries on each side on both sides gives $a=0$. Thus
\begin{align}
\mb{u}^{\circ2}=\mb{P}_{X|Y}\mb{v}^{\circ2}.
\end{align}

The necessity of \eqref{e32} and \eqref{e33} have been shown. To show the further simplification \eqref{e_full1}-\eqref{e_full2} under additional assumptions, notice that
\begin{align}
Q_X=Q_{X|Y}Q_Y,
\end{align}
which, combined with \eqref{e32}, shows that
\begin{align}
D(\mb{u}^{\circ2}\|Q_X)\le D(\mb{v}^{\circ2}\|Q_Y)\label{e_35}
\end{align}
where we abuse the notation by considering $\mb{u}^{\circ2}$ as a probability distribution.
However, the by symmetry we also have $D(\mb{u}^{\circ2}\|Q_X)\ge D(\mb{v}^{\circ2}\|Q_Y)$, so \eqref{e_35} is actually achieved with equality. Denote by
$P_{XY}$ the joint distribution associated with ${\bf Q}_{X|Y}\mb{v}^{\circ2}$.
The necessary and sufficient condition for the data processing inequality \eqref{e_35} to hold with equality is that
$Q_{Y|X}=P_{Y|X}$ holds $P_X$-almost surely.
In the case of indecomposable $Q_{XY}$ and fully supported $Q_X$ and $Q_Y$, it is elementary to show that $\mb{v}^{\circ2}={\bf P}_Y$. The other condition follows from the same reasoning.

\section{An Inequality Related to Conjecture~\ref{conj1} and its Numerical Validation}
\label{app_ineq}
Let
\begin{align}
\mathcal{\bar{S}}_r(X,Y):=(H(X,Y),I(X;Y))-\mathcal{S}_r(X,Y)
\end{align}
be the reflection of $\mathcal{S}_r(X,Y)$ with respect to a point.
The functional $\omega^s_r$ defined in \eqref{e29}-\eqref{e23} can then be represented as
\begin{align}
\omega^s_r(Q_{XY}):=\max_{(\bar{S},\bar{R})
\in\bar{\mathcal{S}}_r(X,Y)}\{s\bar{S}-\bar{R}\}.
\end{align}
Then geometrically, $\omega_r^s(Q_{XY})$ is as illustrated in Figure~\ref{fig_omega}.
Notice that the slope of the supporting line intersecting the upper-right point of $\mathcal{S}_1^s(X,Y)$ (resp.~$\mathcal{S}_{\infty}^s(X,Y)$) is exactly the SDPC $s_1^*(X;Y)$ (resp.~the SSDPC $s_{\infty}^*(X;Y)$), both equal to $(1-2\epsilon)^2$ for a BSS with error probability $\epsilon$.

We need to parameterize the lower set $\mathcal{P}(Q_{XY})$ with two parameters as in \eqref{e_parameterize}, via the bijection $(f,g)\mapsto P_{XY}$.
It can be easily verified that for fixed $g$, the transitional probability $P_{Y|X}$ is also fixed, hence $f$ only controls the marginal $P_X$. Further, the function
\begin{align}
\chi(f,g):=A+\frac{c}{Z}\left(f-\frac{1}{2}\right)\left(g-\frac{1}{2}\right)
\end{align}
is \emph{$XY$-linear} (defined similarly as $XY$-concave with obvious changes) for any real numbers $A$ and $c$. If $\alpha\in[0,\frac{1}{2}]$ is the number that maximizes $sH(X,Y|U)-I(X;Y|U)$ where $U$ is symmetric Bernoulli satisfying $U-X-Y$, then straight forward calculations show that
\begin{align}
s=\frac{(\bar{\epsilon}-\epsilon)(\log(\alpha*\epsilon)-\log(\bar{\alpha}*\epsilon))}{\log\alpha-\log\bar{\alpha}}.
\end{align}
Moreover, $sH(X,Y|U)-I(X;Y|U)$ is equal to $\omega^s_0$ at four points:
\begin{align}
(f,g)=&(\alpha,\frac{1}{2}),\label{p1}\\
(f,g)=&(\bar{\alpha},\frac{1}{2}),\\
(f,g)=&(\frac{1}{2},\alpha),\\
(f,g)=&(\frac{1}{2},\bar{\alpha})\label{p4}.
\end{align}
We can choose a unique $A$ such that $\chi$ and $\omega^s_0$ have the same values at those four points, and a unique $c$ such that the two functions have the same first order derivatives at those four points.
It is an elementary exercise to figure out the values of such $A$ and $c$. If with these values of $A$ and $c$ the $XY$-linear functional $\chi$ dominates $\omega_0^s$, then Conjecture~\ref{conj1} will follow. In other words, Conjecture~\ref{conj1} will be implied by the following conjectured inequality:
\begin{conjecture}\label{conj2}
Suppose $\alpha,\epsilon,f,g\in(0,1)$, and
\begin{align}
s:=\frac{(\bar{\epsilon}-\epsilon)[\log(\alpha*\epsilon)-\log(\bar{\alpha}*\epsilon)]}{\log\alpha-\log\bar{\alpha}}.
\end{align}
(We remind the reader the notations $\bar{\epsilon}:=1-\epsilon$ and $\alpha*\epsilon:=\alpha\bar{\epsilon}+\bar{\alpha}\epsilon$.) Also define
\begin{align}
c&:=\frac{4k\alpha\bar{\alpha}(\bar{\epsilon}-\epsilon)}{\bar{\alpha}-\alpha}\log\frac{\alpha}{\bar{\alpha}}
+4(k+1)\epsilon\bar{\epsilon}\log\frac{\epsilon}{\bar{\epsilon}}
-\frac{4(\epsilon*\alpha)(\epsilon*\bar{\alpha})}{\bar{\alpha}-\alpha}\log\frac{\epsilon*\alpha}{\epsilon*\bar{\alpha}}
\\
&=\frac{4\epsilon\bar{\epsilon}}{\bar{\alpha}-\alpha}
\log\frac{\bar{\alpha}*\epsilon}{\alpha*\epsilon}
-\frac{4\epsilon\bar{\epsilon}(\bar{\epsilon}-\epsilon)\log\frac{\bar{\epsilon}}{\epsilon}
\log\frac{\bar{\alpha}*\epsilon}{\alpha*\epsilon}}{\log\frac{\bar{\alpha}}{\alpha}}
-4\bar{\epsilon}{\epsilon}\log\frac{\bar{\epsilon}}{\epsilon}.
\end{align}
When $\alpha=\frac{1}{2}$ the above are defined via continuity. Then we have
\begin{align}\label{e46}
sH(\hat{X},\hat{Y})-I(\hat{X};\hat{Y})
\le
s[h(\epsilon)+h(\alpha)]-[h(\alpha*\epsilon)-h(\epsilon)]+\frac{c(f-\frac{1}{2})(g-\frac{1}{2})}{f*\bar{g}*\epsilon}.
\end{align}
(remember that $h$ is the binary entropy function, $P_{XY}$ was defined in \eqref{e_parameterize}, and $(\hat{X},\hat{Y})\sim P_{XY}$) and the equality holds at the four points \eqref{p1}-\eqref{p4}.
\end{conjecture}
\begin{rem}\label{rem_1}
By symmetry of the functions involved, we only have to verify for $\alpha,\epsilon,f\in(0,\frac{1}{2})$ and $g\in(0,1)$.
\end{rem}
\begin{rem}
Conjecture~\ref{conj2} is stronger than Conjecture~\ref{conj1}. On the other hand, it can be shown that Conjecture~\ref{conj1} implies the inequality in Conjecture~\ref{conj2} for $(f,g)\in[0,1]\times[\alpha,\bar{\alpha}]\bigcup[\alpha,\bar{\alpha}]\times[0,1]$.
\end{rem}

\begin{rem}
From
\begin{align}
\mathbb{E}[\hat{X}\hat{Y}]=\frac{\bar{\epsilon}(\bar{f}\bar{g}+fg)-\epsilon(\bar{f}g+f\bar{g})}{Z}
\end{align}
we obtain
\begin{align}
\frac{c(f-\frac{1}{2})(g-\frac{1}{2})}{f*\bar{g}*\epsilon}
=\frac{c}{4}\left[\frac{1}{2\bar{\epsilon}}-\frac{1}{2\epsilon}+\left(\frac{1}{2\bar{\epsilon}}+\frac{1}{2\epsilon}
\right)\mathbb{E}[\hat{X}\hat{Y}]\right]
\end{align}
Therefore the conjecture inequality is equivalent to
\begin{align}
&\quad (s+1)H(\hat{X},\hat{Y})
\nonumber\\
&\le
H(\hat{X})+H(\hat{Y})+s[h(\epsilon)+h(\alpha)]
\nonumber\\
&\quad-[h(\alpha*\epsilon)-h(\epsilon)]
\nonumber\\
&\quad+\frac{c}{8\epsilon\bar{\epsilon}}[\epsilon-\bar{\epsilon}+\mathbb{E}[\hat{X}\hat{Y}]]
\end{align}
\end{rem}
Although Conjecture~\ref{conj2} seems elementary, we have not been able to find a full proof. Nevertheless, since it only involves four parameters we can parameterize the space $(0,1)^4$ and verify numerically. We computed the difference between the right hand side of \eqref{e46} and the left hand side. From the choice of $A$ we know that the difference is exactly zero at the four points \eqref{p1}-\eqref{p4}. Using Matlab we computed difference between the right hand side of \eqref{e46} and the left hand side for $f,g,\epsilon,\alpha$ ranging from vectors
\begin{align}
F=&[ss/3:ss:0.5-ss/3]';\\
G=&[ss/3:ss:1-ss/3]';\\
E=&[ss/3:ss:0.5-ss/3]';\\
A=&E;
\end{align}
where the step size $ss:=0.001$. As the result the minimum value of the difference is \verb"-5.841478017444557e-17" with double precision, which is quite small. Moreover negativity of the difference occurs only when $0.496333333333333\le\epsilon<0.5$ and $0.499333333333333\le\alpha<0.5$. If we make $\epsilon$ and $\alpha$ closer to $0.5$, then the magnitude of the difference can further increase, up to about $10^{-9}$ at most; however in this case the image of the left hand side becomes noise-like of the magnitude about $10^{-9}$ as well, so the error is most likely due to the limit of the double precision. In fact, when we use variable precision arithmetic (vpa), the images become smooth and good looking again, and the minimum difference becomes zero.

To visualize what is happening in Conjecture~\ref{conj2}, we plotted $\omega^s_0$, $\chi$ and their difference in Fig.~\ref{fig2}-\ref{fig4} for a particular instance of $\epsilon$ and $\alpha$ (the value of $k$ is then uniquely determined).

\begin{figure}[ht]
  \centering
  \includegraphics[width=3IN]{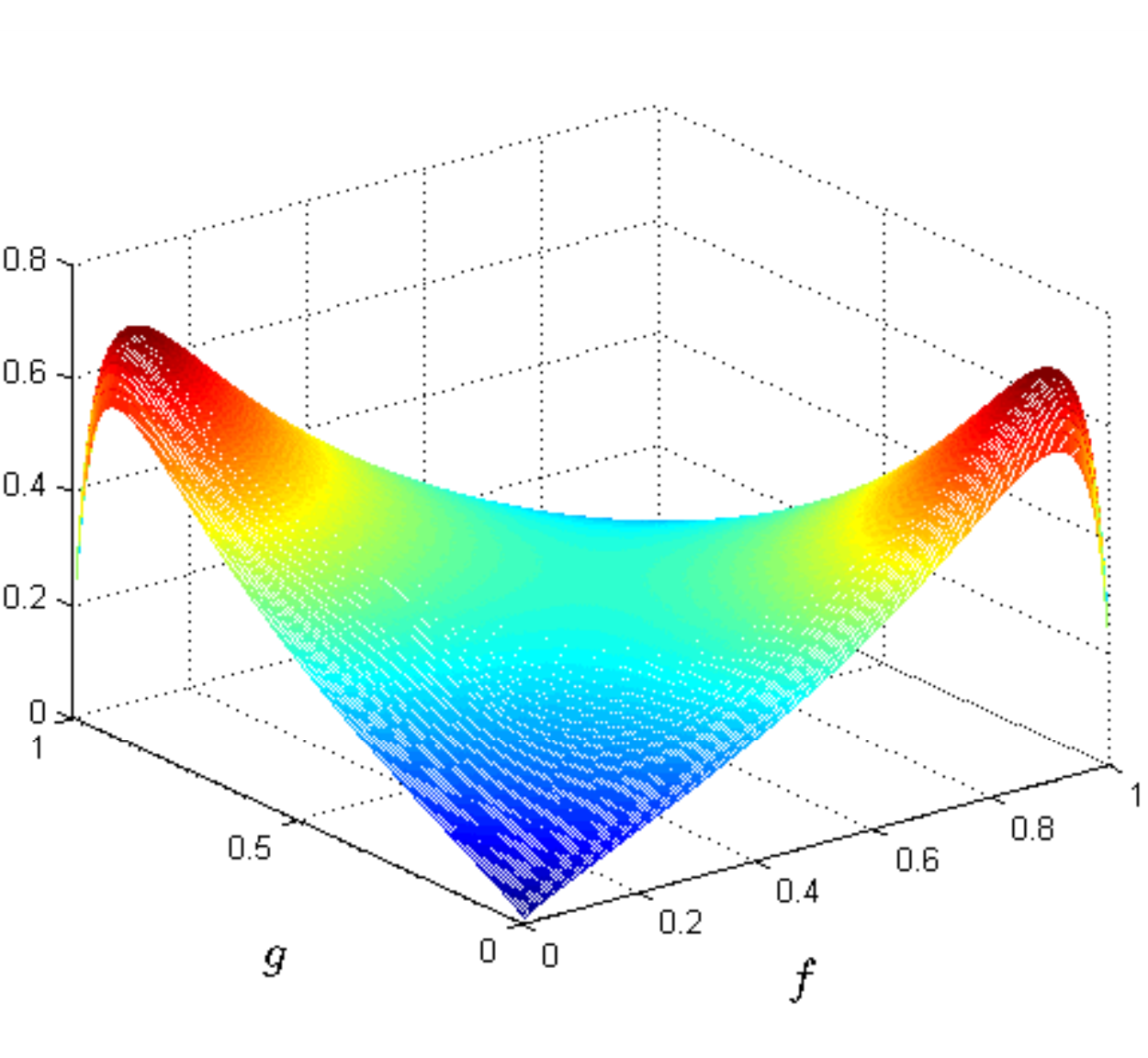}\\
  \caption{Plot of $\omega_0^s$ against $f$ and $g$ when $\alpha=\epsilon=0.11$}\label{fig2}
\end{figure}
\begin{figure}[ht]
  \centering
  \includegraphics[width=3IN]{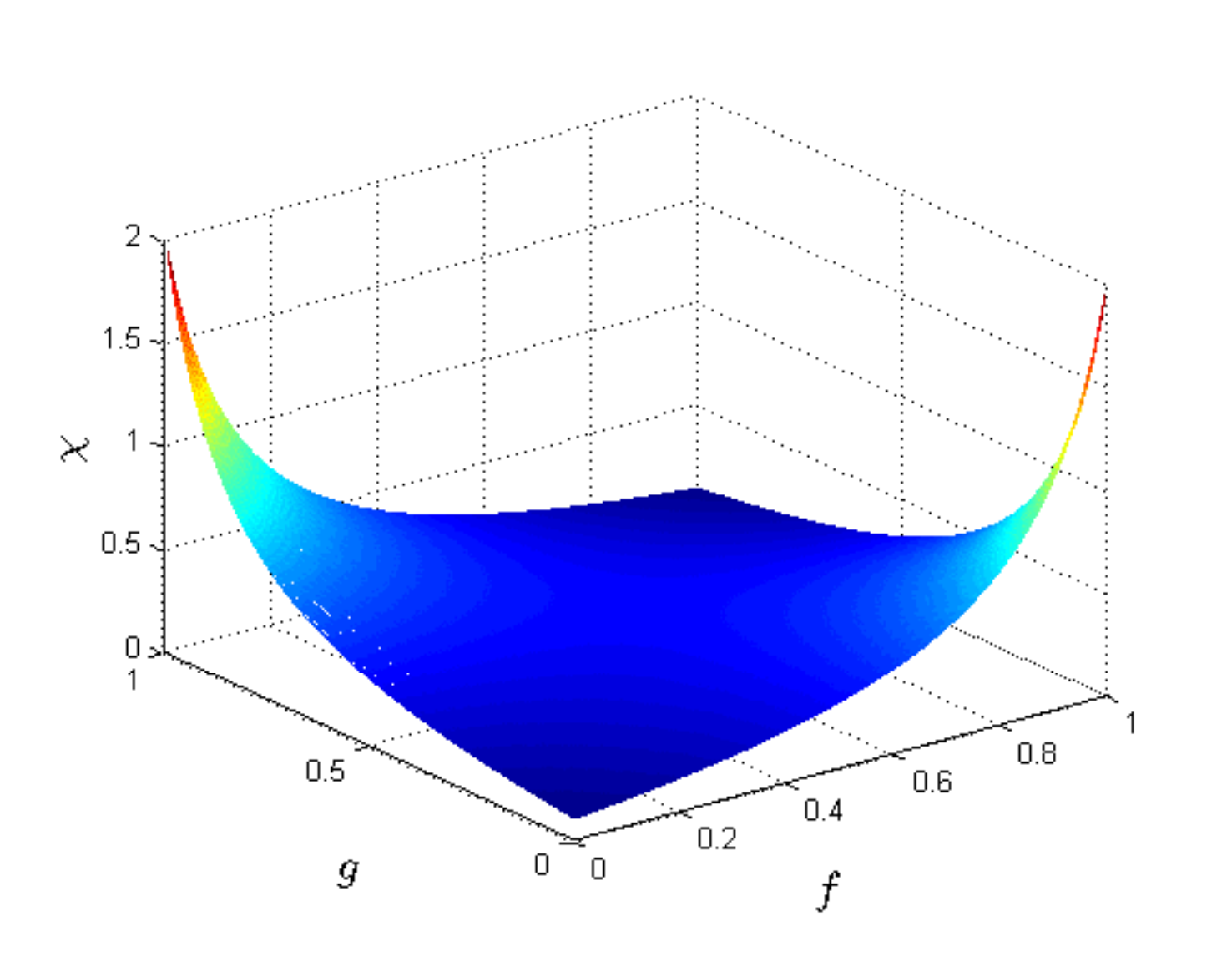}\\
  \caption{Plot of $\chi$ against $f$ and $g$ when $\alpha=\epsilon=0.11$}\label{fig3}
\end{figure}
\begin{figure}[ht]
  \centering
  \includegraphics[width=3IN]{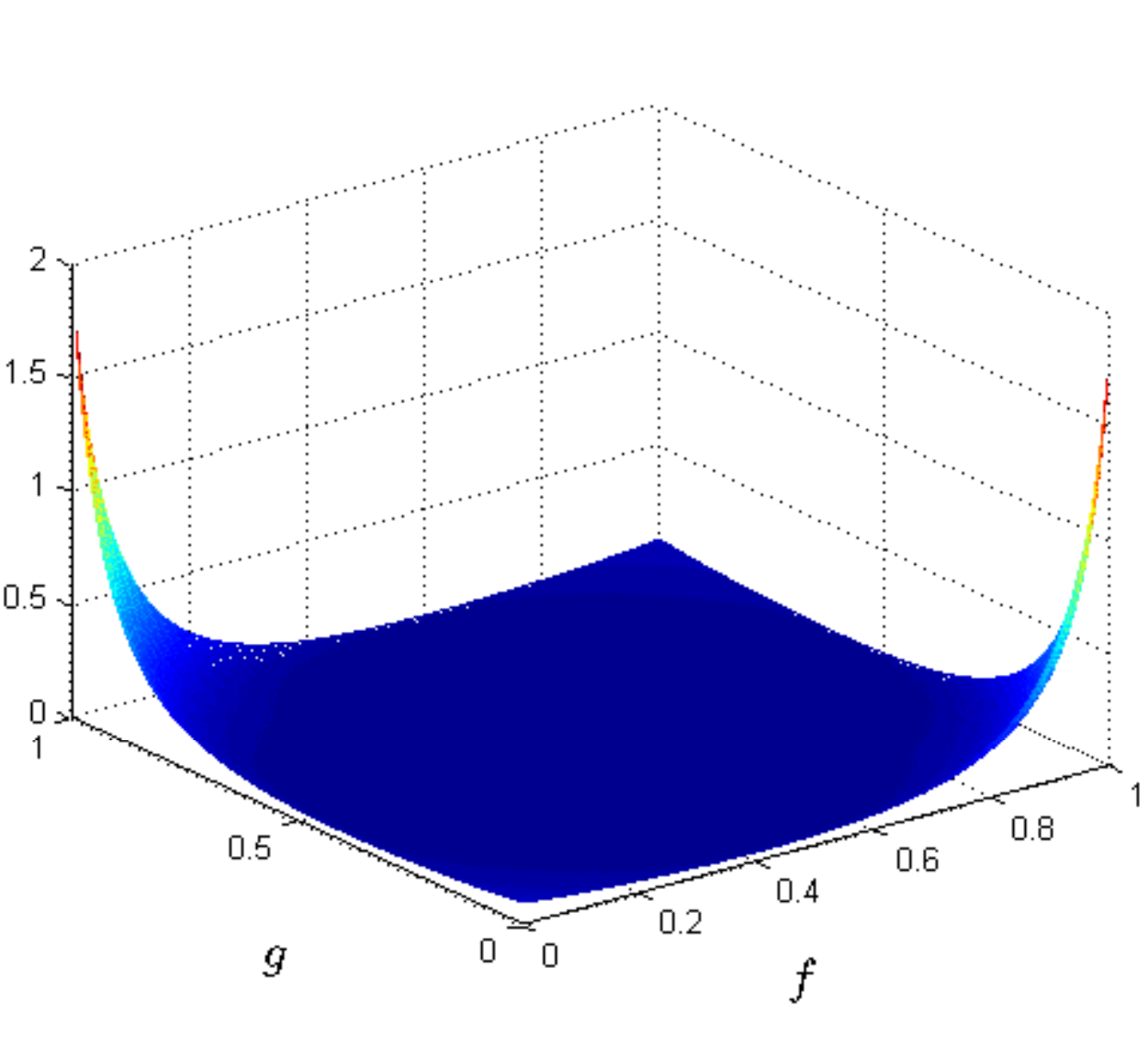}\\
  \caption{Plot of $\chi-\omega_0^s$ against $f$ and $g$ when $\alpha=\epsilon=0.11$}\label{fig4}
\end{figure}

From those numerical results, the inequality is close to failure only in the regime of very small communication rates and very noisy BSS, but in former case, Theorem~\ref{lem2} has guaranteed the validity of the conjecture, while in the latter case, we proved the inequality using Taylor expansion in Appendix~\ref{app_e12}.

\section{Proof of Conjecture~\ref{conj2} for $\epsilon\to\frac{1}{2}$}
\label{app_e12}
We fix $f$, $g$ and $\alpha$ and let $\epsilon\to\frac{1}{2}$. Let $u$ be such that
\begin{align}
\epsilon=\frac{1-u}{2}.
\end{align}
Introduce the notation
\begin{align}
\epsilon(x,y):=\left\{\begin{array}{cc}
                       1-\epsilon & x=y; \\
                       \epsilon & \textrm{otherwise}.
                     \end{array}
\right.
\end{align}
Then,
\begin{align}
\frac{\rd}{\rd \epsilon}\log\epsilon(x,y)=\frac{\log e}{\epsilon(x,y)}(-1)^{x-y+1},
\end{align}
\begin{align}
\frac{\rd^2}{\rd \epsilon^2}\log\epsilon(x,y)=-\frac{\log e}{\epsilon^2(x,y)},
\end{align}
\begin{align}
\frac{\rd}{\rd \epsilon}\log Z=&\frac{\log e}{f*\bar{g}*\epsilon}(f*g-f*\bar{g})
\\
=&-\frac{\log e}{f*\bar{g}*\epsilon}(\bar{f}-f)(\bar{g}-g),
\end{align}
\begin{align}
\frac{\rd^2}{\rd \epsilon^2}\log Z=-\frac{\log e}{(f*\bar{g}*\epsilon)}(\bar{f}-f)^2(\bar{g}-g)^2
\end{align}
\begin{align}
\frac{\rd}{\rd \epsilon}\log\sum_{x'}\epsilon(x',y)f(x')
=\frac{\log e}{\sum_{x'}\epsilon(x',y)f(x')}\sum_{x'}(-1)^{x'-y+1}f(x')
\end{align}
\begin{align}
\frac{\rd^2}{\rd \epsilon^2}\log\sum_{x'}\epsilon(x',y)f(x')
=-\frac{\log e(\sum_{x'}(-1)^{x'-y+1}f(x'))^2}{(\sum_{x'}\epsilon(x',y)f(x'))^2}
\end{align}
In particular,
\begin{align}\label{e60}
\left.\frac{\rd}{\rd \epsilon}\log\epsilon(x,y)\right|_{\epsilon=\frac{1}{2}}=2\log e (-1)^{x-y+1},
\end{align}
\begin{align}
\left.\frac{\rd^2}{\rd \epsilon^2}\log\epsilon(x,y)\right|_{\epsilon=\frac{1}{2}}
=-4\log e,
\end{align}
\begin{align}
\left.\frac{\rd}{\rd \epsilon}\log Z\right|_{\epsilon=\frac{1}{2}}
=-2\log e(\bar{f}-f)(\bar{g}-g),
\end{align}
\begin{align}
\left.\frac{\rd^2}{\rd \epsilon^2}\log Z\right|_{\epsilon=\frac{1}{2}}
=-4\log e(\bar{f}-f)^2(\bar{g}-g)^2,
\end{align}
\begin{align}
\left.\frac{\rd}{\rd \epsilon}\log\sum_{x'}\epsilon(x',y)f(x')\right|_{\epsilon=\frac{1}{2}}
=2\log e\sum_{x'}(-1)^{x'-y+1}f(x'),
\end{align}
\begin{align}\label{e65}
\left.\frac{\rd^2}{\rd \epsilon^2}\log\sum_{x'}\epsilon(x',y)f(x')\right|_{\epsilon=\frac{1}{2}}
=-4\log e(\sum_{x'}(-1)^{x'-y+1}f(x'))^2.
\end{align}
When $\epsilon\to\frac{1}{2}$, we show that both sides of the inequality is of the order of $u^2$. It is easy to compute
\begin{align}
s=\frac{2\log e(\bar{\alpha}-\alpha)u^2}{\log\bar{\alpha}-\log\alpha}+o(u^2),
\end{align}
\begin{align}
c=o(u^2),
\end{align}
\begin{align}\label{e68}
I(\hat{X},\hat{Y})=\sum_{x,y}\frac{\epsilon(x,y)f(x)g(x)}{Z}\log
\frac{\epsilon(x,y)Z}{\sum_{x'}\epsilon(x',y)f(x')\sum_{y'}\epsilon(x,y')g(y')}
\end{align}
where the summations are over $(x,y)\in\{0,1\}^2$, $x'\in\{0,1\}$ and $y'\in\{0,1\}$, respectively. Define
\begin{align}\label{e69}
T(\epsilon):=I(\hat{X};\hat{Y})Z.
\end{align}
Since
\begin{align}
\left.\frac{}{}I(\hat{X};\hat{Y})\right|_{\epsilon=\frac{1}{2}}=0,
\end{align}
and from the minimality of $I(\hat{X};\hat{Y})$ at $\epsilon=\frac{1}{2}$,
\begin{align}
\left.\frac{\rd}{\rd \epsilon}I(\hat{X};\hat{Y})\right|_{\epsilon=\frac{1}{2}}=0,
\end{align}
we have from Leibniz's rule
\begin{align}
&\quad\left.\frac{\rd^2}{\rd \epsilon^2}T\right|_{\epsilon=\frac{1}{2}}
\nonumber\\
&=\left.\frac{\rd^2}{\rd \epsilon^2}I(\hat{X};\hat{Y})\right|_{\epsilon=\frac{1}{2}}\left.\frac{}{}Z\right|_{\epsilon=\frac{1}{2}}
\nonumber\\
&\quad+2\left.\frac{\rd}{\rd \epsilon}I(\hat{X};\hat{Y})\right|_{\epsilon=\frac{1}{2}}
\left.\frac{\rd}{\rd \epsilon}Z\right|_{\epsilon=\frac{1}{2}}
\nonumber\\
&\quad+\left.I(\hat{X};\hat{Y})\right|_{\epsilon=\frac{1}{2}}\left.\frac{\rd^2}{\rd \epsilon^2}Z\right|_{\epsilon=\frac{1}{2}}
\\
&=\frac{1}{2}\left.\frac{\rd^2}{\rd \epsilon^2}I(\hat{X};\hat{Y})\right|_{\epsilon=\frac{1}{2}}.
\end{align}
Thus we obtain
\begin{align}
\left.\frac{\rd^2}{\rd \epsilon^2}I(\hat{X};\hat{Y})\right|_{\epsilon=\frac{1}{2}}
=2\left.\frac{\rd^2}{\rd \epsilon^2}T\right|_{\epsilon=\frac{1}{2}}
\label{e_diff}
\end{align}
which is useful because the differential on the right side of \eqref{e_diff} is easier to compute than the left side.
For $x,y\in\{0,1\}$, $\epsilon\in [0,1]$, define
\begin{align}
\zeta(x,y,\epsilon):=\log\frac{\epsilon(x,y)Z}{\sum_{x'}\epsilon(x',y)f(x')\sum_{y'}\epsilon(x,y')g(y')}.
\end{align}
From \eqref{e68} and \eqref{e69},
\begin{align}
&\left.\frac{\rd^2}{\rd \epsilon^2}T\right|_{\epsilon=\frac{1}{2}}
\nonumber\\
=&\sum_{x,y}
\left.\frac{\rd^2}{\rd \epsilon^2}[\epsilon(x,y)f(x)g(x)]\right|_{\epsilon=\frac{1}{2}}
\left.\zeta(x,y,\epsilon)\right|_{\epsilon=\frac{1}{2}}
\\
&+2\sum_{x,y}
\left.\frac{\rd}{\rd \epsilon}[\epsilon(x,y)f(x)g(x)]\right|_{\epsilon=\frac{1}{2}}
\left.\frac{\rd}{\rd\epsilon}\zeta(x,y,\epsilon)\right|_{\epsilon=\frac{1}{2}}
\\
&+\sum_{x,y}
\left.\frac{}{}\epsilon(x,y)f(x)g(x)\right|_{\epsilon=\frac{1}{2}}
\left.\frac{\rd^2}{\rd\epsilon^2}\zeta(x,y,\epsilon)\right|_{\epsilon=\frac{1}{2}}
\label{e77}
\end{align}
From \eqref{e60}-\eqref{e65}, we see that the first term is zero. The second term is equal to
\begin{align}
&\quad2\sum_{x,y}f(x)g(y)(-1)^{x-y+1}\cdot2\log e
\nonumber\\
&\cdot\left((-1)^{x-y+1}-(\bar{f}-f)(\bar{g}-g)+\sum_{x'}(-1)^{x'-y}f(x')+\sum_{y'}(-1)^{x-y'}g(y')\right)
\\
=&~4\log e\left(1+(\bar{f}-f)^2(\bar{g}-g)^2-(\bar{f}-f)^2-(\bar{g}-g)^2\right)
\end{align}
The third term in \eqref{e77} can be simplified as
\begin{align}
&\quad\sum_{x,y}\frac{f(x)g(y)}{2}\cdot4\log e
\nonumber\\
&\cdot\left[-1-(\bar{f}-f)^2(\bar{g}-g)^2
+(\sum_{x'}(-1)^{x'-y+1}f(x'))^2+(\sum_{y'}(-1)^{x-y'+1}g(y'))^2\right]
\\
&=~2\log e[-1-(\bar{f}-f)^2(\bar{g}-g)^2+(\bar{f}-f)^2+(\bar{g}-g)^2]
\end{align}
Hence
\begin{align}
\left.\frac{\rd^2}{\rd \epsilon^2}T\right|_{\epsilon=\frac{1}{2}}
=4\log e[1-(\bar{f}-f)^2][1-(\bar{g}-g)^2]
\end{align}
Thus we find the left hand side of \eqref{e46} is
\begin{align}
\frac{2\log e(\bar{\alpha}-\alpha)u^2}{\log\bar{\alpha}-\log\alpha}[h(f)+h(g)]-\frac{u^2}{2}\log e[1-(\bar{f}-f)^2][1-(\bar{g}-g)^2]+o(u^2).
\end{align}
The right hand side of \eqref{e46} is
\begin{align}
\frac{2\log e(\bar{\alpha}-\alpha)u^2}{\log\bar{\alpha}-\log\alpha}(1+h(\alpha))
-\frac{\log e}{2}[1-(\bar{\alpha}-\alpha)^2]u^2.
\end{align}
Thus the following inequality implies the validity of \eqref{e46} for fixed $\alpha,f,g\in(0,\frac{1}{2})$ and vanishing $\epsilon-\frac{1}{2}$:
\begin{align}\label{e85}
\frac{2(\bar{\alpha}-\alpha)}{\log\bar{\alpha}-\log\alpha}[h(f)+h(g)]-8f\bar{f}g\bar{g}
\le
\frac{2(\bar{\alpha}-\alpha)}{\log\bar{\alpha}-\log\alpha}(1+h(\alpha))-2\alpha\bar{\alpha}
\end{align}
Note that now we only have to verify the inequality for $g\in(0,\frac{1}{2})$, in contrast to Remark~\ref{rem_1}.
Consider fixed $\alpha$. The values of $f$ and $g$ that maximizes the left hand side of \eqref{e85} must be the solution of the following optimization problem:
\begin{align}\label{opt}
\text{minimize~}\eta(f,g):=16f\bar{f}g\bar{g}\quad\textrm{subject to }\phi(f,g):=h(f)+h(g)=C
\end{align}
for some constant $C$. We solve this minimization problem using Lagrange multiplier method. Define
\begin{align}
L(f,g):=\eta(f,g)-\lambda\phi(f,g).
\end{align}
Suppose $(f^*,g^*)$ is a local minimum, then for some value of $\alpha$, we have
\begin{align}
\left\{\begin{array}{c}
         \partial_1 L(f^*,g^*)=0, \\
         \partial_2 L(f^*,g^*)=0.
       \end{array}
       \right.
\end{align}
for some $\lambda=\lambda^*\neq0$, which implies that
\begin{align}\label{e89}
\left\{\begin{array}{c}
         \partial_1\eta(f^*,g^*)=\lambda^*\partial_1\phi(f^*,g^*),\\
         \partial_2\eta(f^*,g^*)=\lambda^*\partial_2\phi(f^*,g^*).
       \end{array}
       \right.
\end{align}
\begin{enumerate}
  \item If $f^*\neq \frac{1}{2}$ and $g^*\neq \frac{1}{2}$, we can cancel $\lambda^*$ from \eqref{e89} and obtain after rearrangement
      \begin{align}\label{e90}
      \frac{\log\bar{f}^*-\log f^*}{\bar{f}^*-f^*}\bar{f}^*f^*
      =\frac{\log\bar{g}^*-\log g^*}{\bar{g}^*-g^*}\bar{g}^*g^*.
      \end{align}
      It is elementary to check (e.g.~by writing it as Taylor series in terms of $1-2x$ that the function
      \begin{align}
      T(x):=\frac{\log\bar{x}-\log x}{\bar{x}-x}\bar{x}x
      \end{align}
      is monotonically increasing on $(0,\frac{1}{2})$. Thus \eqref{e90} implies that
      \begin{align}\label{e92}
      g^*=f^*.
      \end{align}
      Recall that $(f^*,f^*)$ being a local minimum point implies that the Hessian matrix $[\partial^2_{i,j}L(f^*,f^*)]$ is positive-semidefinite on the orthogonal complement of the span of $\nabla\phi(f^*,f^*)$. In our case, this means that the matrix
      \begin{align}
      \left(
        \begin{array}{cc}
          \frac{\lambda^*\log e}{f^*\bar{f}^*}-32f^*\bar{f}^* & 16(\bar{f}^*-f^*)^2 \\
          16(\bar{f}^*-f^*)^2 & \frac{\lambda^*\log e}{f^*\bar{f}^*}-32f^*\bar{f}^* \\
        \end{array}
      \right)
      \end{align}
      is positive-semidefinite on the span of $(1,-1)^{\top}$, or equivalently,
      \begin{align}\label{e94}
      \frac{\lambda^*\log e}{f^*\bar{f}^*}-32f^*\bar{f}^*
      \ge 16(\bar{f}^*-f^*)^2
      \end{align}
      Substituting \eqref{e92} into \eqref{e89}, we obtain
      \begin{align}
      \lambda^*=\frac{16(\bar{f}^*-f^*)f^*\bar{f}^*}{\log \bar{f}^*-\log f^*},
      \end{align}
      hence \eqref{e94} is equivalent to
      \begin{align}\label{eq96}
      \frac{\log e(\bar{f}^*-f^*)}{\log \bar{f}^*-\log f^*}-2f^*\bar{f}^*\ge (\bar{f}^*-f^*)^2
      \end{align}
      However for any $u^*:=1-2f^*\neq0$, we show that \eqref{eq96} fails:
      \begin{align}
      \textrm{LHS of }\eqref{eq96}
      =&\frac{u^*}{\ln(1+u^*)-\ln(1-u^*)}-2f^*\bar{f}^*
      \\
      =&\frac{u^*}{\sum_{k=1}^{\infty}\frac{(-1)^{k-1}}{k}u^{*k}
      -\sum_{k=1}^{\infty}\frac{(-1)^{k-1}}{k}u^{*k}}-2f^*\bar{f}^*
      \\
      =&\frac{u^*}{\sum_{k=1}^{\infty}\frac{(-1)^{k-1}}{k}u^{*k}
      +\sum_{k=1}^{\infty}\frac{1}{k}u^{*k}}
      -2f^*\bar{f}^*
      \\
      =&\frac{1}{\sum_{l\in 2\mathcal{N}}\frac{2}{l+1}u^{*l}}
      -2f^*\bar{f}^*
      \\
      <&\frac{1}{2}-2f^*\bar{f}^*
      \\
      =&\frac{(\bar{f}^*-f^*)^2}{2}
      \\
      <&\textrm{RHS of }\eqref{eq96}.
      \end{align}
      Therefore, the solution to \eqref{opt} must belong to the following case:

      \item If either $f^*=\frac{1}{2}$ or $g^*=\frac{1}{2}$, by the symmetry of \eqref{e85} we may assume without loss of generality that $g^*=\frac{1}{2}$. The left hand side of \eqref{e85} becomes
          \begin{align}
          \frac{2(\bar{\alpha}-\alpha)}{\log\bar{\alpha}-\log\alpha}(1+h(f))-2f\bar{f}.
          \end{align}
          When viewed as a function of $f$, it is maximized by $f=\alpha$ using Calculus, in which case it agrees with the right hand side of \eqref{e85}. Thus \eqref{e85} is proved.

\end{enumerate}

\bibliographystyle{ieeetrans}
\bibliography{ref2014}

\begin{thebibliography}{10}
\providecommand{\url}[1]{#1}
\csname url@samestyle\endcsname
\providecommand{\newblock}{\relax}
\providecommand{\bibinfo}[2]{#2}
\providecommand{\BIBentrySTDinterwordspacing}{\spaceskip=0pt\relax}
\providecommand{\BIBentryALTinterwordstretchfactor}{4}
\providecommand{\BIBentryALTinterwordspacing}{\spaceskip=\fontdimen2\font plus
\BIBentryALTinterwordstretchfactor\fontdimen3\font minus
  \fontdimen4\font\relax}
\providecommand{\BIBforeignlanguage}[2]{{%
\expandafter\ifx\csname l@#1\endcsname\relax
\typeout{** WARNING: IEEEtranS.bst: No hyphenation pattern has been}%
\typeout{** loaded for the language `#1'. Using the pattern for}%
\typeout{** the default language instead.}%
\else
\language=\csname l@#1\endcsname
\fi
#2}}
\providecommand{\BIBdecl}{\relax}
\BIBdecl

\bibitem{ahlswede1993common}
R.~Ahlswede and I.~Csisz\'{a}r, ``{Common randomness in information theory and
  cryptography. {I}. {S}ecret sharing},'' \emph{IEEE Trans. Inf. Theory},
  vol.~39, no.~4, pp. 1121--1132, Apr.~1993.

\bibitem{ahlswede1998common}
------, ``{Common randomness in information theory and cryptography. Part {II}.
  {CR} capacity},'' \emph{IEEE Trans. Inf. Theory}, vol.~44, no.~1, pp.
  225--240, Jan.~1998.

\bibitem{ahlswede1976spreading}
R.~Ahlswede and P.~G{\'a}cs, ``{Spreading of sets in product spaces and
  hypercontraction of the {Markov} operator},'' \emph{The Annals of
  Probability}, pp. 925--939, 1976.

\bibitem{anan_conj}
V.~Anantharam, A.~A. Gohari, S.~Kamath, and C.~Nair, ``{On hypercontractivity
  and the mutual information between {Boolean} functions},'' \emph{The 51st
  Annual Allerton Conference on Communication, Control, and Computing
  (Allerton)}, pp. 13--19, 2013.

\bibitem{anantharam2013maximal}
V.~Anantharam, A.~Gohari, S.~Kamath, and C.~Nair, ``On maximal correlation,
  hypercontractivity, and the data processing inequality studied by {Erkip} and
  {Cover},'' \emph{arXiv preprint arXiv:1304.6133}, 2013.

\bibitem{boothby2003introduction}
W.~M. Boothby, \emph{An Introduction to Differentiable Manifolds and Riemannian
  Geometry}.\hskip 1em plus 0.5em minus 0.4em\relax Gulf Professional
  Publishing, 2003, vol. 120.

\bibitem{braverman2012coding}
M.~Braverman, ``Coding for interactive computation: Progress and challenges.''
  in \emph{The 50th Annual Allerton Conference on Communication, Control, and
  Computing}, 2012, pp. 1914--1921.

\bibitem{braverman2012interactive}
------, ``Interactive information complexity,'' in \emph{Proceedings of the
  forty-fourth annual ACM symposium on Theory of computing}.\hskip 1em plus
  0.5em minus 0.4em\relax ACM, 2012, pp. 505--524.

\bibitem{braverman2015}
M.~Braverman, R.~Oshman, and O.~Weinstein, ``Information and communication
  complexity,'' \emph{IEEE Information Theory Society Newsletter}, vol.~65,
  no.~3, September 2015.

\bibitem{braverman2014information}
M.~Braverman and A.~Rao, ``Information equals amortized communication,''
  \emph{IEEE Trans. Inf. Theory}, vol.~60, no.~10, pp. 6058--6069, 2014.

\bibitem{choi1994equivalence}
M.-D. Choi, M.~B. Ruskai, and E.~Seneta, ``{Equivalence of certain entropy
  contraction coefficients},'' \emph{Linear algebra and its applications}, vol.
  208, pp. 29--36, 1994.

\bibitem{chou2012separation}
R.~A. Chou and M.~R. Bloch, ``Separation of reliability and secrecy in
  rate-limited secret key generation,'' \emph{IEEE Trans. Inf. Theory},
  vol.~60, no.~8, pp. 4941--4957, Aug.~2014.

\bibitem{courtade2013outer}
T.~A. Courtade, ``Outer bounds for multiterminal source coding via a strong
  data processing inequality,'' in \emph{Proceedings of 2013 IEEE International
  Symposium on Information Theory}, Istanbul, Turkey, July 2013, pp. 559--563.

\bibitem{csiszar1975divergence}
I.~Csisz{\'a}r, ``{$I$-divergence} geometry of probability distributions and
  minimization problems,'' \emph{The Annals of Probability}, pp. 146--158,
  1975.

\bibitem{csiszar1978broadcast}
I.~Csisz{\'a}r and J.~K\"orner, ``Broadcast channels with confidential
  messages,'' \emph{IEEE Trans. Inf. Theory}, vol.~24, no.~3, pp. 339--348,
  1978.

\bibitem{csiszar2011information}
I.~Csisz\'ar and J.~K{\"o}rner, \emph{Information Theory: Coding Theorems for
  Discrete Memoryless Systems}.\hskip 1em plus 0.5em minus 0.4em\relax 2nd ed.
  Cambridge University Press, 2011.

\bibitem{csiszar2004secrecy}
I.~Csisz\'{a}r and P.~Narayan, ``{Secrecy capacities for multiple terminals},''
  \emph{IEEE Trans. Inf. Theory}, vol.~50, no.~12, pp. 3047--3061, Dec.~2004.

\bibitem{csiszar2000common}
------, ``{Common randomness and secret key generation with a helper},''
  \emph{IEEE Trans. Inf. Theory}, vol.~46, no.~2, pp. 344--366, Feb.~2000.

\bibitem{el2011network}
A.~{El Gamal} and Y.-H. Kim, \emph{{Network Information Theory}}.\hskip 1em
  plus 0.5em minus 0.4em\relax Cambridge University Press, 2011.

\bibitem{ganor2014exponential}
A.~Ganor, G.~Kol, and R.~Raz, ``Exponential separation of information and
  communication,'' in \emph{2014 IEEE 55th Annual Symposium on Foundations of
  Computer Science (FOCS)}, pp. 176--185.

\bibitem{gebelein1941}
H.~Gebelein, ``Das statistische problem der korrelation als variations- und
  eigenwert-problem und sein zusammenhang mit der ausgleichungsrechnung,''
  \emph{Zeitschrift fur angew. Math. und Mech.}, vol.~21, pp. 364--379, 1941.

\bibitem{Gohari2012}
A.~Gohari, M.~H. Yassaee, and M.~R. Aref, ``Secure channel simulation,''
  \emph{arXiv:1207.3513}, 2012.

\bibitem{gross1975logarithmic}
L.~Gross, ``{Logarithmic {Sobolev} Inequalities},'' \emph{American Journal of
  Mathematics}, vol.~97, no.~4, pp. 1061--1083, 1975.

\bibitem{han_s}
T.~S. Han, \emph{{Information-Spectrum Method in Information Theory}}.\hskip
  1em plus 0.5em minus 0.4em\relax Springer, 2003.

\bibitem{hayashi2016}
M.~Hayashi, H.~Tyagi, and S.~Watanabe, ``Secret key agreement: General capacity
  and second-order asymptotics,'' \emph{IEEE Trans. Inf. Theory}, vol.~62,
  no.~7, pp. 3796--3810, May 2016.

\bibitem{hirschfeld}
H.~O. Hirschfeld, ``{A connection between correlation and contingency},''
  \emph{Proc. Cambridge Philosophical Soc.}, vol.~31, pp. 520--524, 1935.

\bibitem{kamath2012non}
S.~Kamath and V.~Anantharam, ``Non-interactive simulation of joint
  distributions: The {Hirschfeld-Gebelein-R{\'e}nyi} maximal correlation and
  the hypercontractivity ribbon,'' in \emph{The 50th Annual Allerton Conference
  on Communication, Control, and Computing}, 2012, pp. 1057--1064.

\bibitem{kamath15}
S.~Kamath, ``{Reverse hypercontractivity using information measures},'' in
  \emph{The 53rd Annual Allerton Conference on Communication, Control, and
  Computing}, 2015, pp. 627--633.

\bibitem{kaspi1985two}
A.~H. Kaspi, ``Two-way source coding with a fidelity criterion,'' \emph{IEEE
  Trans. Inf. Theory}, vol.~31, no.~6, pp. 735--740, Nov.~1985.

\bibitem{lancaster}
H.~O. Lancaster, ``Some properties of the bivariate normal distribution
  considered in the form of a contingency table,'' \emph{Biometrika}, vol.~44,
  no. 1--2, pp. 289--292, 1957.

\bibitem{ISIT_lcv_interactive}
J.~Liu, P.~Cuff, and S.~Verd\'u, ``Key generation with limited interaction,''
  in \emph{Proceedings of the 2016 IEEE International Symposium on Information
  Theory}, Barcelona, Spain, July 10--15, 2016, pp. 2918--2922.

\bibitem{Liu}
J.~Liu, P.~Cuff, and S.~Verd\'{u}, ``Key capacity for product sources with
  application to stationary {G}aussian processes,'' \emph{IEEE Trans. Inf.
  Theory}, vol.~62, pp. 1--22, Feb.~2016.

\bibitem{liu2015key}
------, ``{Secret Key Generation with One Communicator and a One-Shot Converse
  via Hypercontractivity},'' in \emph{Proceedings of 2015 IEEE International
  Symposium on Information Theory}, Hong Kong, China, 2015, pp. 710--714.

\bibitem{ma2012interactive}
N.~Ma, P.~Ishwar, and P.~Gupta, ``{Interactive source coding for function
  computation in collocated networks},'' \emph{IEEE Trans. Inf. Theory},
  vol.~58, no.~7, pp. 4289--4305, July, 2012.

\bibitem{maurer1993secret}
U.~M. Maurer, ``{Secret key agreement by public discussion from common
  information},'' \emph{IEEE Trans. Inf. Theory}, vol.~39, no.~3, pp. 733--742,
  Mar.~1993.

\bibitem{milman2007geometrization}
V.~D. Milman, ``Geometrization of probability,'' in \emph{Geometry and Dynamics
  of Groups and Spaces}.\hskip 1em plus 0.5em minus 0.4em\relax Springer, 2007,
  pp. 647--667.

\bibitem{polyanskiy2016lecture}
Y.~Polyanskiy and Y.~Wu, ``Lecture notes on information theory,'' 2016.

\bibitem{polyanskiy2012hypothesis}
Y.~Polyanskiy, ``Hypothesis testing via a comparator,'' in \emph{Proceedings of
  2012 IEEE International Symposium on Information Theory}, Cambridge, MA, July
  2012, pp. 2206--2210.

\bibitem{pw_2015}
Y.~Polyanskiy and Y.~Wu, ``{A note on the strong data-processing inequalities
  in {Bayesian} networks},'' \emph{http://arxiv.org/pdf/1508.06025v1.pdf}.

\bibitem{renyi}
A.~R\'{e}nyi, ``On measures of dependence,'' \emph{Acta Math. Hung.}, vol.~10,
  pp. 441--451, 1959.

\bibitem{rockafellar1970convex}
R.~T. Rockafellar, \emph{{Convex analysis}}.\hskip 1em plus 0.5em minus
  0.4em\relax Princeton university press, 1970.

\bibitem{shannon1949}
C.~E. Shannon, ``Communication in the presence of noise,'' \emph{Proc. IRE},
  vol.~37, pp. 10--21, Jan. 1949.

\bibitem{shannon1949communication}
------, ``{Communication Theory of Secrecy Systems},'' \emph{Bell System
  Technical Journal}, vol.~28, no.~4, pp. 656--715, 1949.

\bibitem{song}
E.~C. Song, P.~Cuff, and H.~V. Poor, ``The likelihood encoder for lossy
  compression,'' \emph{IEEE Trans. Inf. Theory}, vol.~62, no.~4, pp.
  1836--1849, Feb.~2016.

\bibitem{Tyagi2015}
H.~Tyagi and S.~Watanabe, ``Converses for secret key agreement and secure
  computing,'' \emph{IEEE Trans. Inf. Theory}, vol.~61, pp. 4809--4827, July
  2015.

\bibitem{tyagi2013common}
H.~Tyagi, ``{Common information and secret key capacity},'' \emph{IEEE Trans.
  Inf. Theory}, vol.~59, no.~9, pp. 5627--5640, Sept.~2013.

\bibitem{valiant}
G.~Valiant and P.~Valiant, ``A {CLT} and tight lower bounds for estimating
  entropy,'' \emph{Proc. Electron. Colloq. Comput. Complex. (ECCC)}, vol.~17,
  p. 179, 2010.

\bibitem{verdu1990channel}
S.~Verd\'{u}, ``{On channel capacity per unit cost},'' \emph{IEEE Trans. Inf.
  Theory}, vol.~36, no.~5, pp. 1019--1030, May~1990.

\bibitem{neumann1928}
J.~von Neumann, ``Zur theorie der gesellschaftsspiele,'' \emph{Math. Annalen.},
  vol. 100, pp. 295--320, 1928.

\bibitem{watanabe2010}
S.~Watanabe and Y.~Oohama, ``Secret key agreement from correlated gaussian
  sources by rate limited public communication,'' \emph{IEICE Transactions on
  Fundamentals of Electronics, Communications and Computer Sciences}, vol.
  E93-A, no.~11, pp. 1976--1983, 2010.

\bibitem{witsenhausen1975sequences}
H.~S. Witsenhausen, ``On sequences of pairs of dependent random variables,''
  \emph{SIAM Journal on Applied Mathematics}, vol.~28, no.~1, pp. 100--113.

\bibitem{witsenhausen1975conditional}
H.~S. Witsenhausen and A.~D. Wyner, ``A conditional entropy bound for a pair of
  discrete random variables,'' \emph{IEEE Trans. Inf. Theory}, vol.~21, no.~5,
  pp. 493--501, 1975.

\bibitem{wyner1973}
A.~D. Wyner and J.~Ziv, ``A theorem on the entropy of certain binary sequences
  and applications: {Part I},'' \emph{IEEE Trans. Inf. Theory}, vol.~19.

\bibitem{yao1979}
A.~C.-C. Yao, ``Some complexity questions related to distributive computing,''
  in \emph{STOC}, 1979, pp. 209--213.

\end{thebibliography}
\end{document}